\documentclass[runningheads]{llncs}
\usepackage{graphicx}
\usepackage[utf8]{inputenc}
\usepackage{ltl}
\usepackage{todonotes}
\usepackage{stmaryrd}
\usepackage{amssymb}
\usepackage{xcolor}
\usepackage{tikz}
\usepackage{pgfplots}
\usepackage{caption}
\usepackage{subcaption}
\usepackage{tabularx}
\usepackage{listings}
\usetikzlibrary{arrows,automata,shapes,shapes.multipart,positioning,calc}

\usepackage[ruled,vlined,linesnumbered]{algorithm2e}
\usepackage{mathtools}

\usepackage{hyperref}
\usepackage[capitalise,nameinlink]{cleveref} 

\usepackage[stable]{footmisc} 

\usepackage{etoolbox}
\makeatletter
\patchcmd\algocf@Vline{\vrule}{\vrule \kern-0.4pt}{}{}
\patchcmd\algocf@Vsline{\vrule}{\vrule \kern-0.4pt}{}{}
\makeatother

\DeclarePairedDelimiter\abs{\lvert}{\rvert}

\newcommand{\ie}{i.e.\@\xspace}

\newcommand{\eg}{e.g.\@\xspace}

\newcommand{\Next}{\LTLnext}
\newcommand{\Globally}{\LTLsquare}
\newcommand{\Eventually}{\LTLdiamond}
\newcommand{\Until}{\LTLuntil}

\newcommand{\values}{\mathcal{V}}
\newcommand{\functions}{\mathcal{F}}
\newcommand{\predicates}{\mathcal{P}}
\newcommand{\funcTerms}{\mathcal{T}_F}
\newcommand{\predTerms}{\mathcal{T}_P}
\newcommand{\allModels}[1]{\mathit{Models}(#1)}

\newcommand{\free}[1]{\operatorname{free}(#1)}

\newcommand{\true}{\mathit{true}}
\newcommand{\false}{\mathit{false}}

\newcommand{\update}[2]{\llbracket #1 \leftarrowtail #2 \rrbracket}
\newcommand{\predTerm}{\tau_p}
\newcommand{\funcTerm}{\tau_f}
\newcommand{\cell}[1]{{\normalfont{\texttt{#1}}}}
\newcommand{\cells}{\mathbb{C}}
\newcommand{\funcLiterals}{\Sigma_F}
\newcommand{\predLiterals}{\Sigma_P}
\newcommand{\assign}{{\langle \cdot \rangle}}
\newcommand{\finitary}[2]{\operatorname{fin}_{#1}(#2)}

\newcommand{\allUpdates}[1]{\operatorname{updates}(#1)}
\newcommand{\allPredicates}[1]{\operatorname{predicates}(#1)}
\newcommand{\allCells}[1]{\operatorname{cells}(#1)}

\newcommand{\upd}[2]{\update{\mathit{#1}}{#2}}
\newcommand{\updcell}[2]{\update{\mathit{#1}}{\mathit{#2}}}
\newcommand{\updconst}[2]{\update{\mathit{#1}}{\mathit{#2}()}}
\newcommand{\updappl}[3]{\update{\mathit{#1}}{#2(\mathit{#3})}}
\newcommand{\spec}[1]{\varphi_\mathit{#1}}
\newcommand{\predconst}[1]{p(\mathit{#1}())}
\newcommand{\predcell}[1]{p(\mathit{#1})}

\newcommand{\computation}{\varsigma}
\newcommand{\etaFunction}[3]{\eta(#1,#2,#3)}
\newcommand{\recursiveAssign}[1]{\chi_{\assign}(#1)}
\newcommand{\recursiveAssignFin}[1]{\chi_{\assign_\mathit{fin}}(#1)}
\newcommand{\TSLSatisfaction}[2]{\models_{#1,#2}}

\newcommand{\guards}{\mathcal{G}}
\newcommand{\updateTerms}{\mathcal{U}}
\newcommand{\fresh}{\bullet}
\newcommand{\atPos}[2]{{{#1}_{#2}}}
\newcommand{\projection}[2]{\pi_{#2}(#1)}

\newcommand{\symbolicLanguage}[1]{\mathcal{L}(#1)}
\newcommand{\theoryLanguage}[2]{\mathcal{L}_{#1}(#2)}
\newcommand{\symbolicExecution}{(\computation,\mathcal{V},\assign)}

\newcommand{\constraintTrace}{\varrho}

\newcommand{\executionEffect}[1]{\operatorname{effect}(#1)}

\newcommand{\compBSA}{\mathcal{R}}
\newcommand{\accBSA}[1]{\operatorname{accepting}(#1)}

\newcommand{\tu}{T_U}
\newcommand{\te}{T_E}
\newcommand{\tp}{T_\mathbb{N}}

\newcommand{\pref}{\mathit{pref}\!}
\newcommand{\rec}{\mathit{rec}}

\newcommand{\sat}{{\normalfont{SAT}}}
\newcommand{\unsat}{{\normalfont{UNSAT}}}
\newcommand{\while}{\textsc{While}}

\newcommand{\goto}{\textsc{Goto}}
\newcommand{\gotoNSC}{\normalfont{GOTO}}
\newcommand{\gotoifeq}{\textsc{GotoIfEq}}
\newcommand{\inc}{\textsc{Inc}}
\newcommand{\res}{\textsc{Reset}}
\newcommand{\dec}{\textsc{Dec}}

\newcommand{\encInc}{f}
\newcommand{\encEq}{\mathrel{\widehat{=}}}
\newcommand{\encRes}{z}
\newcommand{\tslgoto}[1]{\varphi^{#1}_\mathcal{G}}

\newcommand{\depth}[1]{\operatorname{depth}(#1)}
\newcommand{\maxDepth}[2]{\operatorname{maxDepth}_{#1}(#2)}

\newenvironment{proofsketch}{%
  \proof}{\endproof}

\begin{document}
\title{Temporal Stream Logic modulo Theories \\ (Full Version)\thanks{This work was partially supported by the German Research Foundation~(DFG) as part of the Collaborative Research Center ``Foundations of Perspicuous Software Systems'' (TRR 248, 389792660), and by the European Research Council (ERC) Grant OSARES (No. 683300). Philippe Heim and Noemi Passing carried out this work as PhD candidates at Saarland University, Germany.}\footnote[3]{This is an extended version of~\cite{FinalVersion}.}}
\titlerunning{Temporal Stream Logic modulo Theories (Full Version)}
%
\author{Bernd Finkbeiner \and 
Philippe Heim\and 
Noemi Passing
}
\authorrunning{B.\ Finkbeiner et al.}
%
\institute{CISPA Helmholtz Center for Information Security, Saarbrücken, Germany
\email{\{finkbeiner, philippe.heim, noemi.passing\}@cispa.de}
}
\maketitle              
\begin{abstract}
	Temporal stream logic (TSL) extends LTL with updates and predicates over arbitrary function terms. This allows for specifying data-intensive systems for which LTL is not expressive enough. In the semantics of TSL, functions and predicates are left uninterpreted. In this paper, we extend TSL with first-order theories, enabling us to specify systems using interpreted functions and predicates such as incrementation or equality. We investigate the satisfiability problem of TSL modulo the standard underlying theory of uninterpreted functions as well as with respect to Presburger arithmetic and the theory of equality: For all three theories, TSL satisfiability is highly undecidable. Nevertheless, we identify three fragments of TSL for which the satisfiability problem is (semi-)decidable in the theory of uninterpreted functions. Despite the high undecidability, we present an algorithm -- which is not guaranteed to terminate -- for checking the satisfiability of a TSL formula in the theory of uninterpreted functions and evaluate it: It scales well and is able to validate assumptions in a real-world system design.
\end{abstract}

\setcounter{footnote}{0}


\section{Introduction}

Linear-time temporal logic (LTL)~\cite{Pnueli77} is one of the standard specification languages to describe properties of reactive systems. The success of LTL is largely due to its ability to abstract from the detailed data manipulations and to focus on the change of control over time. In data-intensive applications, such as smartphone apps, LTL is, however, often not expressive enough to capture the relevant properties.
When specifying a music player app, for instance, we would like to state that if the user leaves the app, the track that is currently playing will be stored and will resume playing once the user returns to the app.

To specify data-intensive systems, extensions of LTL such as Constraint LTL (CLTL)~\cite{DemriD07} and, more recently, Temporal Stream Logic (TSL)~\cite{FinkbeinerKPS19} have been proposed. In CLTL, the atomic propositions of LTL are replaced with atomic constraints over a concrete domain~$D$ and an interpretation for relations. 
Relating variables with the equality relation, such as $x = y$, denoting that the value of~$x$ is equal to the value of $y$, allows for specifying assignment-like statements.
In this paper, however, we focus on the logic TSL to specify data-intensive systems.

TSL extends LTL with updates and predicates over arbitrary function terms. An update $\update{x}{f(y)}$ denotes that the result of applying function $f$ to variable~$y$ is assigned to variable~$x$. For the music player app, for instance, the update $\update{\mathit{paused}}{\mathit{track}(\mathit{current})}$ specifies that the track that is currently playing, obtained by applying function $\mathit{track}$ to variable $\mathit{current}$, is stored in variable $\mathit{paused}$.
Updates are the main characteristic of TSL that differentiates it from other first-order extensions of LTL:
They allow for specifying the evolution of variables over time. Thus, programs can be represented in TSL and therefore, for instance, the model checking problem can be encoded.

In the semantics of TSL, functions and predicates are left uninterpreted, \ie, a system satisfies a TSL formula if the formula evaluates to $\true$ for all possible interpretations of the function and predicate symbols. This semantics has proven especially useful in the synthesis of reactive programs~\cite{FinkbeinerKPS19,GeierH0F19}, where the synthesis algorithm builds a control structure, while the implementation of the functions and predicates is either done manually or provided by some library. One exemplary success story of TSL-based specification and synthesis of a reactive system is the arcade game Syntroids~\cite{GeierH0F19} realized on an FPGA.

In this paper, we define and investigate the satisfiability problem of TSL modulo the standard underlying theory of uninterpreted functions and with respect to other first-order theories such as the theory of equality and Presburger arithmetic.
Intuitively, a TSL formula~$\varphi$ is satisfiable in a theory $T$ if there is an execution satisfying~$\varphi$ that matches the function applications and predicate constraints of an interpretation in~$T$. TSL validity in $T$ is dual: A TSL formula~$\varphi$ is valid in a theory $T$ if, and only if, $\neg \varphi$ is unsatisfiable in $T$.

For LTL, satisfiability is decidable~\cite{SistlaC85} and efficient algorithms for checking the satisfiability of an LTL formula have been implemented in tools like Aalta~\cite{LiZPVH13}. Satisfiability checking has numerous applications in the specification and analysis of reactive systems, such as identifying inconsistent system requirements during the design process, comparing different formalizations of the same requirements, and various types of vacuity checking.
TSL satisfiability checking extends these applications to data-intensive systems.

We present an algorithm for checking the satisfiability of a TSL formula in the theory of uninterpreted functions. It is based on \emph{Büchi stream automata}~(BSAs), a new kind of $\omega$-automata that we introduce in this paper. BSAs can handle the predicates and updates occurring in TSL formulas. Similar to the relationship between LTL formulas and nondeterministic Büchi automata, BSAs are an automaton representation of TSL formulas, \ie, there exists an equivalent BSA for every TSL formula.
Given a TSL formula $\varphi$, our algorithm constructs an equivalent BSA $\mathcal{B}_\varphi$ and then tries to prove satisfiability and unsatisfiability in parallel: For proving satisfiability, it searches for a lasso that ensures consistency of the function terms in an accepting run of $\mathcal{B}_\varphi$. If such a lasso is found, $\varphi$ is satisfiable.
For proving unsatisfiability, the algorithm discards inconsistent runs of $\mathcal{B}_\varphi$. If no accepting run is left, $\varphi$ is unsatisfiable.

The algorithm does not always terminate. In fact, we show that TSL satisfiability is neither semi-decidable nor co-semi-decidable in the theory of uninterpreted functions. Thus, no complete algorithm exists. The undecidability result extends to the theory of equality and Presburger arithmetic.
There exist, however, (semi-)decidable fragments of TSL in the theory of uninterpreted functions: For satisfiable formulas with a single variable as well as satisfiable reachability formulas, our algorithm is guaranteed to terminate. For slightly more restricted reachability formulas, satisfiability is decidable.

We have implemented the algorithm and evaluated it, clearly illustrating its applicability:~It terminates within one second on many randomly generated formulas and scales particularly well for satisfiable formulas. Moreover, it is able to check realistic benchmarks for consistency and to (in-)validate their assumptions.
Most notably, we successfully validate the assumptions of a Syntroids module.

A preliminary version of this paper has been published on arXiv~\cite{arxiv}. This already lead to further research on TSL modulo theories: Maderbacher and Bloem show that the synthesis problem for TSL modulo theories is undecidable in general and present a synthesis procedure for TSL modulo theories based on a counter-example guided LTL synthesis loop~\cite{MaderbacherB21}.

\section{Preliminaries}

We assume time to be discrete.
A \emph{value} can be of arbitrary type and we denote the set of all values by~$\values$.
The Boolean values are denoted by $\mathbb{B} \subseteq \values$.
Given $n$ values, an $n$-ary function $f: \values^n \rightarrow \values$ computes a new value. An $n$-ary predicate $p: \values^n \rightarrow \mathbb{B}$ determines whether a property over~$n$ values is satisfied.
The sets of all functions and predicates are denoted by~$\functions$ and $\predicates \subseteq \functions$, respectively.
Constants are both functions of arity zero and values.
Starting from $0$, we denote the $i$-th position of an infinite word $\sigma$ by $\atPos{\sigma}{i}$ and the $i$-th component of a tuple $t$ by $\projection{t}{i}$.

To argue about functions and predicates, we use a term based notation. 
\emph{Function terms~$\funcTerm$} are constructed from variables and functions, recursively applied to a set of function terms.
\emph{Predicate terms~$\predTerm$} are constructed by applying a predicate to function terms.
The sets of all function and predicate terms are denoted by $\funcTerms$ and $\predTerms \subseteq \funcTerms$, respectively.
Given sets $\Sigma_F$, $\Sigma_P$ of function and predicate symbols with $\Sigma_P \subseteq \Sigma_F$, a set $V$ of variables, and a set $\mathcal{V}$ of values, let $\assign : V \cup \Sigma_F \to \mathcal{V} \cup \mathcal{F}$ be an \emph{assignment function} assigning a concrete function (predicate) to each function (predicate) symbol and an initial value to each variable. 
We require $\langle v \rangle \in \mathcal{V}$, $\langle f \rangle \in \mathcal{F}$, and $\langle p \rangle \in \mathcal{P}$ for $v \in V$, $f \in \Sigma_F$, $p \in \Sigma_P$.
The evaluation $\chi_{\assign} : \funcTerms \to \mathcal{V} \cup \mathbb{B}$ of function terms is defined by $\recursiveAssign{v} := \langle v \rangle$ for $v \in V$, and by $\recursiveAssign{f(\tau_0, \dots, \tau_n)} := \langle f \rangle (\recursiveAssign{\tau_0}, \dots, \recursiveAssign{\tau_n})$ for $f \in \Sigma_F \cup \Sigma_P$.

Functions and predicates are not tied to a specific interpretation. To restrict the possible interpretations, we utilize \emph{first-order theories}.
A first-order theory~$T$ is a tuple $(\Sigma_F, \Sigma_P, \mathcal{A})$, where $\Sigma_F$ and $\Sigma_P$ are sets of function and predicate symbols, respectively, and $\mathcal{A}$ is a set of closed first-order logic formulas over~$\Sigma_F$,~$\Sigma_P$, and a set of variables $V$. For an introduction to first-order logic, we refer to \Cref{app:fol}.
The elements of~$\mathcal{A}$ are called the \emph{axioms} of~$T$ and $\Sigma_F \cup \Sigma_P$ is called the \emph{signature} of~$T$.
A \emph{model} $\mathcal{M}$ for a theory $T=(\Sigma_F, \Sigma_P, \mathcal{A})$ is a tuple $(\mathcal{V}, \assign)$, where~$\mathcal{V}$ is a set of values and $\assign$ is an assignment function as introduced above. Furthermore, $(\mathcal{V}, \assign)$ is required to entail $\varphi_A$ for each axiom $\varphi_A \in \mathcal{A}$. The set of all models of a theory $T$ is denoted by $\allModels{T}$.

In the remainder of this paper, we focus on the following three theories:
The \emph{theory of uninterpreted functions $\tu$} is a theory without any axioms, \ie, every symbol is uninterpreted. It allows for arbitrarily many function and predicate symbols. 
The \emph{theory of equality $\te$} additionally includes equality, \ie, its axioms enforce the equality symbol $=$ to indeed represent equality. 
The \emph{theory of Presburger arithmetic $\tp$} implements the idea of numbers. Its axioms define the constants $0$ and $1$ as well as equality and addition.

\section{Temporal Stream Logic modulo Theories}

In this section, we introduce \emph{Temporal Stream Logic modulo theories}, an extension of the recently introduced logic Temporal Stream Logic (TSL)~\cite{FinkbeinerKPS19} with first-order theories. First, we recap the main idea of TSL as well as its syntax and semantics. Afterwards, we extend TSL with first-order theories and define the basic notions of satisfiability and validity for TSL formulas modulo theories.

\subsection{Temporal Stream Logic}
Temporal Stream Logic (TSL)~\cite{FinkbeinerKPS19} is a temporal logic that separates temporal control and pure data. Data is represented as infinite streams of arbitrary type. TSL allows for checks and manipulations of streams on an abstract level: It focuses on the control flow and abstracts away concrete implementation details. The temporal structure of the data is expressed by temporal operators as in~LTL~\cite{Pnueli77}.
TSL is especially designed for reactive synthesis and thus distinguishes between uncontrollable input streams and controllable out\-put streams, so-called \emph{cells}.
In this paper, this distinction is not necessary since we consider TSL independent of its usage in synthesis. 
Thus, we use the notions of streams and cells as synonyms. The finite set of all cells is denoted by $\cells$.

In TSL, we use functions $f\in\functions$ to modify cells and predicates $p\in\predicates$ to perform checks on cells. 
The cells $\cell{c} \in \cells$ serve as variables for function terms.
The sets of all function and predicate terms over $\funcLiterals$, $\predLiterals$, and $\cells$ are denoted by $\funcTerms$ and $\predTerms$.
TSL formulas are built according to the following grammar:
\[ \varphi, \psi := \true \mid \neg \varphi \mid \varphi \land \psi \mid \Next \varphi \mid \varphi \Until \psi \mid \predTerm \mid \update{\cell{c}}{\funcTerm} \]
where $\cell{c} \in \cells$, $\predTerm \in \predTerms$, and $\funcTerm \in \funcTerms$.
An \emph{update} $\update{\cell{c}}{\tau_f}$ denotes that the value of the function term $\tau_f$ is assigned to cell $\cell{c}$. 
The value of $\tau_f$ may depend on the value of the cells occurring in $\tau_f$.
The temporal operators $\Next \varphi$ and $\varphi \Until \psi$ are similar to the ones in LTL.
We define $\Eventually \varphi = \true \Until \varphi$ and $\Globally \varphi = \neg \Eventually \neg \varphi$. 

Since functions and predicates are represented symbolically, they are not tied to a specific implementation. 
To assign an interpretation to them, we use an assignment function $\assign: \cells \cup \Sigma_F \to \mathcal{V} \cup \mathcal{F}$, where $\mathcal{V}$ is a set of values. We require $\langle \cell{c} \rangle \in \mathcal{V}$, $\langle f \rangle \in \mathcal{F}$ and $\langle p \rangle \in \mathcal{P}$ for $\cell{c} \in \cells$, $f \in \Sigma_F$, and $p \in \Sigma_P$
Note that~$\assign$ also assigns an initial value to each cell.
Terms can be compared syntactically with the equivalence relation $\equiv$.
The set of all assignments of cells $\cell{c} \in \cells$ to function terms $\funcTerm \in \funcTerms$ is denoted by $\mathcal{C}$.
A \emph{computation} $\computation \in \mathcal{C}^\omega$ is an infinite sequence of assignments of cells to function terms, capturing the behavior of cells over time.
The satisfaction of a TSL formula $\varphi$ with respect to $\computation$, a set of values $\mathcal{V}$, an assignment function $\assign$, and a time step $t$ is defined by:\footnote{Note that we use a slightly different, but equivalent, definition than~\cite{FinkbeinerKPS19}: Instead of evaluating the function and predicate symbols on the fly, we construct the whole term first and then evaluate it recursively using the evaluation function~$\chi_\assign$.}
\begin{alignat*}{3}
	&\computation, t \TSLSatisfaction{\mathcal{V}}{\assign} \neg \varphi ~ &:\Leftrightarrow & ~ \computation, t \not\TSLSatisfaction{\mathcal{V}}{\assign} \varphi\\
	&\computation, t \TSLSatisfaction{\mathcal{V}}{\assign} \varphi \land \psi ~ &:\Leftrightarrow & ~ \computation, t \TSLSatisfaction{\mathcal{V}}{\assign} \varphi \land \computation, t \TSLSatisfaction{\mathcal{V}}{\assign} \psi\\
	&\computation, t \TSLSatisfaction{\mathcal{V}}{\assign} \Next \varphi ~ &:\Leftrightarrow & ~ \computation, t+1 \TSLSatisfaction{\mathcal{V}}{\assign} \varphi\\
	&\computation, t \TSLSatisfaction{\mathcal{V}}{\assign} \varphi \Until \psi ~ &:\Leftrightarrow & ~ \exists t'' \geq t. \forall t \leq t' < t''. ~ \computation, t' \TSLSatisfaction{\mathcal{V}}{\assign} \varphi ~ \land ~ \computation, t'' \TSLSatisfaction{\mathcal{V}}{\assign} \psi\\
    &\computation, t \TSLSatisfaction{\mathcal{V}}{\assign} \update{\cell{c}}{\tau} ~ &:\Leftrightarrow & ~ \atPos{\computation}{t}(\cell{c}) \equiv \tau\\
	&\computation, t \TSLSatisfaction{\mathcal{V}}{\assign} p(\tau_0,\dots,\tau_{m}) ~ &:\Leftrightarrow & ~ \recursiveAssign{\etaFunction{\computation}{t}{p(\tau_0,\dots,\tau_{m})}},
\end{alignat*}
where $\eta: \mathcal{C}^\omega \times \mathbb{N} \times \funcTerms \rightarrow \funcTerms$ is a symbolic evaluation function defined by
\begin{align*}
	\etaFunction{\computation}{t}{\cell{c}} &= \begin{cases}
		\cell{c} & \text{if } t = 0\\
		\etaFunction{\computation}{t-1}{\atPos{\computation}{t-1}(\cell{c})} & \text{if } t > 0
	\end{cases}\\
	\etaFunction{\computation}{t}{f(\tau_0, \dots, \tau_{m})} &= f(\etaFunction{\computation}{t}{\tau_0}, \dots, \etaFunction{\computation}{t}{\tau_{m}})
\end{align*}

We call $\symbolicExecution$ an \emph{execution}. 
The satisfaction of a predicate depends on the current and the past steps in the computation. For updates, the satisfaction depends solely on the current step.
While updates are only checked syntactically, the satisfaction of predicates depends on the given assignment~$\assign$.
An execution $\symbolicExecution$ \emph{satisfies} a TSL formula~$\varphi$, denoted $\computation \TSLSatisfaction{\mathcal{V}}{\assign} \varphi$, if $\computation, 0 \TSLSatisfaction{\mathcal{V}}{\assign} \varphi$ holds.

\begin{example}\label{ex:tsl}
	Suppose that we have a single cell $\cell{x}$, \ie, $\cells = \{ \cell{x} \}$.
	Consider the computation $\computation = (\{ \lambda \cell{c}.f(\cell{x}) \})^\omega$, \ie, $f(\cell{x})$ is assigned to cell $\cell{x}$ in every time step. Let $\mathcal{V}=\mathbb{N}$ be the set of values and let $\assign$ be an assignment function such that the initial value of $\cell{x}$ is $1$, function $f$ corresponds to incrementation, and predicate~$p$ determines whether its argument is even ($\true)$ or odd ($\false$).
	Consider the TSL formula $\varphi := \update{\cell{x}}{f(\cell{x})} \land \neg p(\cell{x}) \land \Next p(\cell{x})$. By the semantics of TSL, we have $\computation, 0 \TSLSatisfaction{\mathcal{V}}{\assign} \varphi$ if, and only if, $(\atPos{\computation}{0}(\cell{x}) = f(\cell{x})) \land (\neg \langle p \rangle ( \langle \cell{x} \rangle)) \land (\langle p \rangle(\langle f \rangle(\langle \cell{x} \rangle)))$ holds.
	The first conjunct clearly holds by construction of $\computation$. Since $1$ is odd and $1+1 = 2$ is even, the other two conjuncts hold as well for the chosen assignment function. Hence, $(\computation, \mathcal{V}, \assign)$ satisfies $\varphi$ for $\computation = (\{ \lambda \cell{c}.f(\cell{x}) \})^\omega$, $\mathcal{V} = \mathbb{N}$ and $\assign$.
\end{example}

A computation $\computation$ is called \emph{finitary} with respect to~$\varphi$, denoted $\finitary{\varphi}{\computation}$, if for all cells $\cell{c} \in \cells$ and for all points in time~$t$, either $\atPos{\computation}{t}(\cell{c}) \equiv \cell{c}$ holds, or there is an update $\update{\cell{c}}{\tau}$ in $\varphi$ such that $\atPos{\computation}{t}(\cell{c}) \equiv \tau$, \ie, a finitary computation only contains updates occurring in~$\varphi$ and self-updates. For $\computation$ and $\varphi$ from \Cref{ex:tsl}, for instance, $\computation$ is finitary with respect to $\varphi$.

\subsection{Extending TSL with Theories}
In this paper, we extend TSL with first-order theories. That is,  we restrict the possible interpretations of predicate and function symbols to a theory.
Hence, we define the notions of satisfiability and validity of a TSL formula \emph{modulo a theory} $T$.
Intuitively, a TSL formula $\varphi$ is satisfiable in a theory $T$ if there exists an execution satisfying $\varphi$ whose domain and assignment function represent a model in $T$, \ie, that entail all axioms of $T$. Formally:

\begin{definition}[TSL Satisfiability]
	Let $T = (\Sigma_F, \Sigma_P, \mathcal{A})$ be a theory and let~$\varphi$ be a TSL formula over $\Sigma_F$, $\Sigma_P$, and $\cells$. 
    We call $\varphi$ \emph{satisfiable in~$T$} if, and only if, there exists an execution $\symbolicExecution$, such that $\computation \TSLSatisfaction{\mathcal{V}}{\assign} \varphi$ and $(\mathcal{V}, \assign) \in \allModels{T}$ hold.
	If additionally $\finitary{\varphi}{\computation}$ holds, then $\varphi$ is called \emph{finitary satisfiable in~$T$}.
\end{definition}

Intuitively, a formula $\varphi$ is valid in a theory $T$, if for all executions and all matching models of the theory the formula is satisfied. Formally:

\begin{definition}[TSL Validity]
	Let $T = (\Sigma_F, \Sigma_P, \mathcal{A})$ be a theory and let $\varphi$ be a TSL formula over $\Sigma_F$, $\Sigma_P$, and $\cells$. 
    The formula $\varphi$ is called \emph{valid in~$T$} if, and only if, for all
    executions $\symbolicExecution$ with $(\mathcal{V},\assign) \in \allModels{T}$, we have $\computation \TSLSatisfaction{\mathcal{V}}{\assign} \varphi$.
    If $\computation \TSLSatisfaction{\mathcal{V}}{\assign} \varphi$ holds for all executions $\symbolicExecution$ with both $(\mathcal{V},\assign) \in \allModels{T}$ and $\finitary{\varphi}{\computation}$, then $\varphi$ is called \emph{finitary valid in $T$}.
\end{definition}

It follows directly from their definitions that (finitary) TSL satisfiability and (finitary) TSL validity are dual. Thus, we focus on TSL satisfiability in the remainder of this paper as the results can easily be extended to TSL validity.

\begin{theorem}[Duality of TSL Satisfiability and Validity]\label{thm:duality}
	Let $\varphi$ be a TSL formula over $\Sigma_F$, $\Sigma_P$, and $\cells$ and let $T = (\Sigma_F, \Sigma_P, \mathcal{A})$ be a theory. Then, $\varphi$ is (finitary) satisfiable in $T$ if, and only if, $\neg \varphi$ is not (finitary) valid in $T$. 
\end{theorem}

\section{TSL modulo $\tu$ Satisfiability Checking}

In this section, we investigate the satisfiability of TSL modulo the theory of uninterpreted functions $\tu$.
Since $\tu$ has no axioms, there are no restrictions on how a model for $\tu$ evaluates the function and predicate symbols. The only condition is that the evaluated symbols are indeed functions. Therefore, we have $(\computation, \mathcal{V}, \assign) \in \allModels{\tu}$ for all executions. Thus, finding some execution satisfying a TSL formula $\varphi$ is sufficient for showing that $\varphi$ is satisfied in~$\tu$:

\begin{lemma}\label{lem:canonical_model_tu}
    Let $\varphi$ be a TSL formula over $\Sigma_F$, $\Sigma_P$, and $\cells$. If there exists an execution $(\computation, \mathcal{V}, \assign)$ with $\computation \TSLSatisfaction{\mathcal{V}}{\assign} \varphi$, then $\varphi$ is satisfiable in $\tu$. If additionally $\finitary{\varphi}{\computation}$ holds, then $\varphi$ is finitary satisfiable in $\tu$.
\end{lemma}

In the following, we introduce an (incomplete) algorithm for checking the satisfiability of a TSL formula $\varphi$ in the theory of uninterpreted functions. By \Cref{lem:canonical_model_tu}, it suffices to find an execution satisfying $\varphi$ to prove its satisfiability in~$\tu$. To search for such an execution, we introduce \emph{Büchi stream automata}~(BSAs), a new kind of $\omega$-automata that reads executions and allows for dealing with predicates and updates. BSAs are, similar to the connection between LTL and Büchi automata, an automaton representation for TSL.
Then, we present the algorithm for checking satisfiability in $\tu$ based on~BSAs.

\subsection{Büchi Stream Automata}
Intuitively, a \emph{Büchi stream automaton} (BSA) is an $\omega$-automaton with Büchi acceptance condition that reads infinite executions instead of infinite words. Furthermore, it is able to deal with predicates and updates occurring in TSL formulas.
To do so, the transitions of a BSA are labeled with \emph{guards} and \emph{update terms}. Intuitively, the former define which predicates need to hold when taking the transition. The latter define how the corresponding cells are updated when taking the transition. Formally, a BSA is defined as follows:

\begin{definition}[Büchi Stream Automaton]
	Let $\Sigma_F$, $\Sigma_P$ be sets of function and predicate symbols, respectively, and let $\cells$ be a finite set of cells.
	A \emph{Büchi Stream automaton} $\mathcal{B}$ over $\Sigma_F$, $\Sigma_P$, and $\cells$ is a tuple $(Q, Q_0, F, \fresh, \guards, \updateTerms, \delta)$, where $Q$ is a finite set of states, $Q_0 \subseteq Q$ is a set of initial states, $F \subseteq Q$ is a set of accepting states, $\fresh$ is a fresh term symbol such that $\fresh \not\in \cells \cup \Sigma_F \cup \Sigma_P$, $\guards \subseteq \predTerms$ is a finite set of predicate terms over $\Sigma_F$, $\Sigma_P$, and $\cells$, called \emph{guards}, $\updateTerms \subseteq \funcTerms \cup \{\fresh\}$ is a finite set of function terms over $\Sigma_F$, $\Sigma_P$, and $\cells$, called \emph{update terms}, and $\delta \subseteq Q \times (\guards \to \mathbb{B}) \times (\cells \to \updateTerms) \times Q$ is a finite transition relation.
\end{definition}

Note that by requiring the update terms $\updateTerms$ to be a \emph{finite} set of function terms, not all executions can be read by a BSA: Non-finitary executions contain updates with function terms that do not occur in the given TSL formula. Thus, they may require infinitely many update terms.
Therefore, we introduce the fresh term symbol $\fresh \not \in \cells \cup \Sigma_F \cup \Sigma_P$. If a transition in a BSA assigns $\fresh$ to a cell $\cell{c} \in \cells$, then any function term can be assigned to $\cell{c}$. This allows for reading non-finitary executions while maintaining finite representability of BSAs.

\tikzstyle{label}=[align=center, above, scale = 0.9]
\tikzstyle{aState}=[draw, color=black, rectangle, minimum size=25]
\tikzstyle{trans}=[color=black, ->, >=stealth]

\colorlet{guard}{black!40!red}
\colorlet{updateTerm}{blue!70!black}

\begin{figure}[t]
\centering
\begin{subfigure}[b]{0.36\textwidth}
\centering
\begin{tikzpicture}[scale=0.9, every node/.style={scale=1},initial text=]
    \node[state,initial] 	(A) at (0,0) {$q_0$};
    \node[state,accepting]	(B) at (3.2,0) {$q_1$};
    \path
        (A) edge [trans, loop below] node[label, below] {$\textcolor{guard}{\neg p(\cell{x})}$:\\$\textcolor{updateTerm}{\update{\cell{x}}{f(\cell{x})}}$} (A)
        (A) edge [trans, loop above] node[label, above] {$\textcolor{guard}{p(\cell{x})}$:\\$\textcolor{updateTerm}{\update{\cell{x}}{f(\cell{x})}}$} (A)
        (A) edge [trans, bend left=20] node[label,above] {$\textcolor{guard}{p(\cell{x})}$:\\$\textcolor{updateTerm}{\update{\cell{x}}{f(\cell{x})}}$} (B)
        (B) edge [trans, bend left=20] node[label, below] {$\textcolor{guard}{\neg p(\cell{x})}$:\\$\textcolor{updateTerm}{\update{\cell{x}}{f(\cell{x})}}$} (A);
\end{tikzpicture}
\caption{BSA $\mathcal{B}_1$}\label{fig:bsa_example_1}
\end{subfigure}
\hfill
\begin{subfigure}[b]{0.44\textwidth}
\centering
\begin{tikzpicture}[scale=0.9, every node/.style={scale=1},initial text=]
    \node[state,initial] (A) at (0,0) {$q_0$};
    \node[state,accepting] (B) at (3.2,0) {$q_1$};
    \path
        (A) edge [trans, loop below] node[label, below] {$\textcolor{guard}{p(\cell{x}), p(f(\cell{x}))}$:\\$\textcolor{updateTerm}{\update{\cell{x}}{f(\cell{x})}}$} (A)
        (A) edge [trans] node[label] {$\textcolor{guard}{p(\cell{x}), \neg p(f(\cell{x}))}$:\\$\textcolor{updateTerm}{\update{\cell{x}}{f(\cell{x})}}$} (B)
        (B) edge [trans, loop above] node[label] {$\textcolor{guard}{p(\cell{x}), p(f(\cell{x}))}$:\\$\textcolor{updateTerm}{\update{\cell{x}}{f(\cell{x})}}$} (B)
        (B) edge [trans, loop below] node[label, below] {$\textcolor{guard}{p(\cell{x}), \neg p(f(\cell{x}))}$:\\$\textcolor{updateTerm}{\update{\cell{x}}{f(\cell{x})}}$} (B);
\end{tikzpicture}
\caption{BSA $\mathcal{B}_2$}\label{fig:bsa_example_2}
\end{subfigure}
\hfill
\begin{subfigure}[b]{0.15\textwidth}
\centering
\begin{tikzpicture}[scale=0.9, every node/.style={scale=1},initial text=]
    \node[state,accepting,initial] (A) at (0,0) {$q_0$};
    \path
        (A) edge [trans, loop below] node[label, below] {$\textcolor{guard}{p(\cell{x})}$:\\$\textcolor{updateTerm}{\update{\cell{x}}{\fresh}}$} (A);
\end{tikzpicture}
\caption{BSA $\mathcal{B}_3$}\label{fig:bsa_example_3}
\end{subfigure}
\caption{Three exemplary Büchi stream automata. Accepting states are marked with double circles. Guards are highlighted in red, update terms in blue.}\label{fig:bsa_example}
\end{figure}
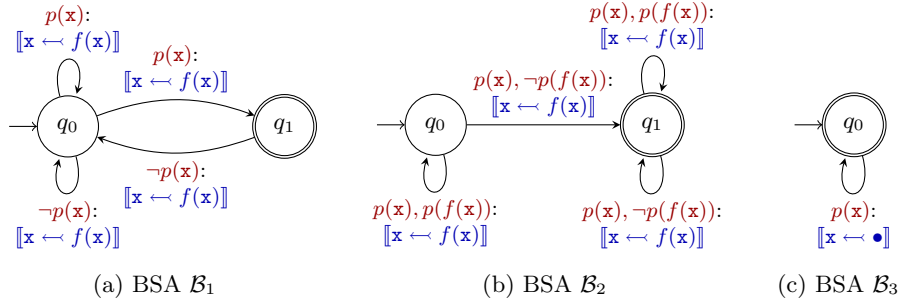

\begin{example}\label{ex:bsa}
	Consider the three BSAs depicted in \Cref{fig:bsa_example}.	
	If~$\mathcal{B}_1$ is in state $q_0$ and $p(\cell{x})$ holds, then cell~$\cell{x}$ is updated with $f(\cell{x})$ and $\mathcal{B}_1$ chooses nondeterministically to either stay in~$q_0$ or to move to the accepting state $q_1$. In contrast, $\mathcal{B}_2$ is deterministic. Yet, it is incomplete: In both~$q_0$ and~$q_1$, no guard is satisfied if $\neg p(\cell{x})$ holds. Hence, $\mathcal{B}_2$ gets stuck, preventing satisfaction of the Büchi winning condition for any execution containing $\neg p(\cell{x})$. The BSA $\mathcal{B}_3$ makes use of the fresh term symbol $\fresh$: If $p(\cell{x})$ holds, any function term can be assigned to~$\cell{x}$.
\end{example}

Given sets $\Sigma_F$, $\Sigma_P$, $\cells$ and a BSA $\mathcal{B} = (Q, Q_0, F, \fresh, \guards, \updateTerms, \delta)$ over $\Sigma_F$, $\Sigma_P$,~$\cells$, an infinite word $c \in (Q \times (\guards \to \mathbb{B}) \times (\cells \to \updateTerms) \times Q)^\omega$ is called \emph{run of $\mathcal{B}$} if, and only if, the first state of $c$ is an initial state, \ie, $\projection{\atPos{c}{0}}{1} \in Q_0$, and both $\atPos{c}{t} \in \delta$ and $\projection{\atPos{c}{t}}{4} = \projection{\atPos{c}{t+1}}{1}$ hold for all points in time $t \in \mathbb{N}_0$.
Intuitively, a run $c$ is an infinite sequence of tuples $(q,g,u,q')$ encoding transitions in the BSA: $q$ is the source state, $q'$ is the target state, $g$ determines which predicate terms hold, and $u$ defines which updates are performed when taking the transition. 
A run $c$ is called \emph{accepting} if it contains infinitely many accepting states, \ie, for all points in time $t \in \mathbb{N}_0$, there exists a $t' > t$ such that $\projection{\atPos{c}{t'}}{1} \in F$ holds.

\begin{example}\label{ex:run}
	Let $g_1(p(\cell{x})) = \true$, $g_2(p(\cell{x})) = \false$, and $u(\cell{x}) = f(\cell{x})$.
	The infinite word $c = (q_0,g_1,u,q_1) (q_1,g_2,u,q_0) (q_0,g_1,u,q_1) (q_1,g_2,u,q_0) \dots$ is a run of BSA~$\mathcal{B}_1$ from \Cref{fig:bsa_example_1}. It is accepting as it visits $q_1$ infinitely often.
\end{example}

The characteristics of a BSA are its predicates and updates. Thus, it is not sufficient to solely consider accepting runs since the constraints produced by the predicates might be inconsistent. Therefore, we define the \emph{execution of a BSA} that only permits consistent accepting runs.
Intuitively, given a run $c$ of a BSA~$\mathcal{B}$, an \emph{execution of c} consists of a computation $\computation \in \mathcal{C}^\omega$, a domain $\mathcal{V}$, and an assignment $\assign$ such that the updates in $\computation$ match the updates in $c$ and such that the recursive evaluation of a predicate term using $\assign$ matches its truth value in~$\computation$.
To capture the constraints accumulated in $\computation$ as well as their truth values, we define the \emph{constraint trace $\constraintTrace: (\predTerm \times \mathbb{B})^\omega$ of $\computation$ and $c$}: Formally, $\constraintTrace$ for $\computation$ and $c$ is defined by $\atPos{\constraintTrace}{t} := \emptyset$ if $t=0$, and $\atPos{\constraintTrace}{t} := \atPos{\constraintTrace}{t-1} \cup \{ (\etaFunction{\computation}{t-1}{\predTerm}, \projection{\atPos{c}{t-1}}{2}(\predTerm) ) \mid \predTerm \in \guards \}$ if $t>0$.
As an example, reconsider the computation $\computation$ from \Cref{ex:tsl} and the run~$c$ of BSA $\mathcal{B}_1$ from \Cref{ex:run}. The constraint trace of $\computation$ and~$c$ is given by $\constraintTrace = \emptyset \{ (p(\cell{x}), \true) \} \{ (p(\cell{x}), \true), (p(f(\cell{x})),\false) \} \dots$.
A constraint trace~$\constraintTrace$ is called \emph{consistent} if there is no predicate term $\predTerm\in\predTerms$ such that both $(\predTerm,\true)$ and $(\predTerm,\false)$ occur in $\bigcup_{t \in \mathbb{N}_0} \atPos{\constraintTrace}{t}$.
$\constraintTrace$ from the example above is consistent.
Using constraint traces, we now formally define the execution of a BSA:

\begin{definition}[Execution of a BSA]
	Let $\Sigma_F$ and $\Sigma_P$ be sets of function and predicate symbols, respectively, and let $\cells$ be a finite set of cells. Let~$\mathcal{B}$ be a BSA over $\Sigma_F$, $\Sigma_P$, and~$\cells$ and let $c$ be a run of $\mathcal{B}$.
	Let $\computation \in \mathcal{C}^\omega$ be an infinite computation and let $\assign: \cells \cup \Sigma_F \to \mathcal{V} \cup \mathcal{F}$ be an assignment function. Let $\constraintTrace$ be the constraint trace of $\computation$ and $c$. We call $\symbolicExecution$ \emph{execution for $c$} if (1) for all points in time~$t\in\mathbb{N}_0$ and all cells $\cell{c} \in \cells$, we have either $\projection{\atPos{c}{t}}{3}(\cell{c}) = \atPos{\computation}{t}(\cell{c})$ or $\projection{\atPos{c}{t}}{3}(\cell{c}) = \fresh$, and (2) for all $(\predTerm, b) \in \bigcup_{t \in \mathbb{N}_0} \atPos{\constraintTrace}{t}$, we have $\recursiveAssign{\predTerm} = b$.
\end{definition}

Note that the second requirement can only be fulfilled if the constraint trace is consistent.
Consider the computation $\computation$ and the assignment function $\assign$ from \Cref{ex:tsl}, the run~$c$ of~$\mathcal{B}_1$ from \Cref{ex:run}, and the constraint trace $\constraintTrace$ of $\computation$ and $c$ given above. Then, $(\computation, \mathbb{N}, \assign)$ is an execution for~$c$: 
Since in both $\computation$ and $c$, cell $\cell{x}$ is always updated with $f(\cell{x})$, the updates in $\computation$ and $c$ coincide at every point in time. Furthermore, by construction of $\assign$, the constraints of~$\constraintTrace$ match the truth values obtained by recursively evaluating $\assign$ for all predicate terms.

We define two languages of a BSA $\mathcal{B}$: The \emph{symbolic language} $\symbolicLanguage{\mathcal{B}}$ is the set of all executions that have a respective accepting run, \ie, $\symbolicExecution \in \mathcal{L}(\mathcal{B})$ if, and only if, there exists an accepting run $c$ such that $\symbolicExecution$ is an execution for~$c$. The \emph{language $\theoryLanguage{T}{\mathcal{B}}$ in a theory $T$} is the set of all executions whose domain and assignment function additionally form a model in $T$, \ie, $\symbolicExecution \in \theoryLanguage{T}{\mathcal{B}}$ if, and only if, $\symbolicExecution \in \symbolicLanguage{\mathcal{B}}$ and $(\mathcal{V},\assign)\in\allModels{T}$.

We call a BSA $\mathcal{B} = (Q, Q_0, F, \fresh, \guards, \updateTerms, \delta)$ \emph{finitary} if $\fresh \not\in \updateTerms$ holds. 
Hence, every run~$c$ of a finitary BSA, has a unique computation~$\computation$ and thus a unique constraint trace~$\constraintTrace$.
Therefore, for a finite prefix $c_p$ of $c$, we can compute its \emph{execution effect} $\executionEffect{c_p} := (\lambda \cell{c}.~\etaFunction{\computation}{|c_p|}{\cell{c}},\atPos{\constraintTrace}{|c_p|})$ from $c_p$ itself, \ie, without considering $\computation$ and $\constraintTrace$ explicitly. Intuitively,~$c_p$'s execution effect consists of the function terms assigned to the cells during the execution of~$c_p$ as well as the constraints and their truth values on the transitions taken with~$c_p$ in the BSA. 
The BSAs $\mathcal{B}_1$ and $\mathcal{B}_2$, depicted in \Cref{fig:bsa_example}, are finitary while~$\mathcal{B}_3$ is not.
Since $\mathcal{B}_1$ is finitary, consider the prefix $c_p = (q_0,g_1,u,q_1) (q_1,g_2,u,q_0)$ of the run $c$ of $\mathcal{B}_1$ presented in \Cref{ex:run}. Its execution effect is given by $\executionEffect{c_p} = (\lambda \cell{c}.f(f(\cell{x})), \{ (p(\cell{x}), \true), (p(f(\cell{x})),\false) \})$.

An LTL formula $\varphi$ can be translated into a nondeterministic Büchi automaton~(NBA)~$\mathcal{A}_\varphi$ with $\mathcal{L}(\varphi) = \mathcal{L}(\mathcal{A}_\varphi)$~\cite{VardiW94}.
An analogous relation exists between TSL formulas and BSAs:
A TSL formula $\varphi$ can be translated into an equivalent BSA~$\mathcal{B}_\varphi$: First, we approximate $\varphi$ by an LTL formula $\varphi_\mathit{LTL}$, similarly to the approximation described in~\cite{FinkbeinerKPS19}. The main idea of the approximation is to represent every function and predicate term as well as every update occurring in $\varphi_\mathit{LTL}$ by an atomic proposition and to add conjuncts that ensure that exactly one update is performed for every cell in every time step. Second, we build an equivalent NBA $\mathcal{A}_{\varphi_\mathit{LTL}}$ from $\varphi_\mathit{LTL}$. Third, we construct a BSA $\mathcal{B}_\varphi$ from $\mathcal{A}_{\varphi_\mathit{LTL}}$ by, intuitively, translating the atomic propositions back into predicate terms and updates and by dividing them into guards and update terms, while maintaining the structure of the NBA $\mathcal{A}_{\varphi_\mathit{LTL}}$.
The full construction of an equivalent BSA $\mathcal{B}_\varphi$ from a TSL formula $\varphi$ is given in \Cref{app:partial_decider}.

\begin{theorem}[TSL to BSA Translation]\label{thm:tsl_to_bsa}
	Given a TSL formula $\varphi$, there exists an equivalent (finitary) Büchi stream automaton~$\mathcal{B}$ such that for all theories~$T$, $\theoryLanguage{T}{\mathcal{B}} \neq \emptyset$ holds if, and only if, $\varphi$ is (finitary) satisfiable in~$T$.	
\end{theorem}

For instance, the TSL formula $\varphi_1 := \Globally \update{\cell{x}}{f(\cell{x})} \land \Globally\Eventually (p(\cell{x}) \land \Next \neg p(\cell{x}))$ is finitary satisfiable in a theory $T$ if, and only if, $\theoryLanguage{T}{\mathcal{B}_1} \neq \emptyset$ holds for the BSA~$\mathcal{B}_1$ from \Cref{fig:bsa_example_1}. Analogously, $\varphi_2 := \Globally (\update{\cell{x}}{f(\cell{x})} \land p(\cell{x})) \land \Eventually \neg p(f(\cell{x}))$, and $\varphi_3 := \Globally p(\cell{x})$ correspond to the BSAs $\mathcal{B}_2$ and~$\mathcal{B}_3$ from \Cref{fig:bsa_example_2} and \Cref{fig:bsa_example_3}.


\subsection{An Algorithm for TSL modulo $\tu$ Satisfiability Checking}\label{sec:partialdecider}

Utilizing BSAs, we present an algorithm for checking the satisfiability of a TSL formula in the theory of uninterpreted functions $\tu$ in the following.
First, recall that finitary computations only perform self-updates or updates that occur in the given TSL formula. Since there are only finitely many cells, the behavior of finitary computations is thus restricted to a finite set of possibilities in each step. Hence, reasoning with finitary computations is easier than reasoning with non-finitary ones.
In the algorithm, we make use of the fact that satisfiability can be reduced to finitary satisfiability in the theory of uninterpreted functions, enabling us to focus on finitary computations.
The main idea of the reduction is to introduce a new cell for each cell of a given TSL formula. The new cells then capture the values that are constructed by the non-finitary parts of a computation. The proof is given in \Cref{app:partial_decider}.

\begin{lemma}\label{lem:reduction_to_finitary}
    Let $\varphi$ be a TSL formula. Then, there is a TSL formula~$\varphi_\mathit{fin}$ such that $\varphi$ is satisfiable in $\tu$ if, and only if, $\varphi \land \varphi_\mathit{fin}$ is finitary satisfiable in $\tu$.
\end{lemma}

\begin{algorithm}[t]
	\SetVlineSkip{1.5pt}
    \SetKwFunction{SMT}{SMT} 
    \SetKwProg{Fn}{Function}{}{}
    \KwIn{$\varphi$: TSL Formula}
    \KwOut{\sat, \unsat}
    $\mathcal{B}$ := Finitary BSA for $\varphi$ as defined in \Cref{thm:tsl_to_bsa}\;
    $\compBSA$ := Set of runs of $\mathcal{B}$\; 
    \Fn{SatSearch}{
        \For{$\mathit{pref}.\mathit{rec}^\omega \in \{ c \mid c \in \compBSA \land \accBSA{c} \}$}{
        ($v_p$, \_):=$\executionEffect{\mathit{pref}}$\;
        ($v_r$, $P$):=$\executionEffect{\mathit{pref}.\mathit{rec}}$\;
        \If{\SMT $\!\left(\!\!\left(\bigwedge\limits_{(t_p, v) \in P} \begin{cases} t_p&\text{if } v = \true\\\neg t_p&\text{if } v = \false\end{cases} \right)\! 
            \land \bigwedge\limits_{\cell{c} \in \cells} v_p(\cell{c}) = v_r(\cell{c})\!\right)\! =\!$ \sat}{
            \Return~\sat
        }
    }
    }
    \Fn{UnsatSearch}{
    \For{$n \in \mathbb{N}_0$}{
        \For{$c \in \{ c \mid c \in \operatorname{finiteSubwords}(\compBSA) \land |c| = n\}$}{
            (\_, $P$):=$\executionEffect{c}$\;
            \If{$\exists t_p.~(t_p, \true), (t_p, \false) \in P$}{
                $\compBSA$ := $\compBSA \setminus \{c' \mid \exists m \in \mathbb{N}_0.~\forall 0 \leq i < n.~\atPos{c'}{i + m} = \atPos{c}{i} \}$
            }
            \If{$\{ c \mid c \in \compBSA \land \accBSA{c} \}$ = $\emptyset$}{
                \Return~\unsat
            }
        }
    }}
    \Return{parallel(SatSearch, UnsatSearch)}
    \caption{Algorithm for Checking TSL modulo $\tu$ Satisfiability}
    \label{alg:TSL_TU_SAT_idea}
\end{algorithm}

\Cref{alg:TSL_TU_SAT_idea} shows the algorithm for checking TSL modulo $\tu$ satisfiability. It directly works on Büchi stream automata. First, an equivalent BSA~$\mathcal{B}$ is generated for the input formula $\varphi$. Then, in parallel, $\mathit{SatSearch}$ tries to prove that~$\varphi$ is satisfiable in $\tu$ while $\mathit{UnsatSearch}$ tries to prove unsatisfiability of $\varphi$.

$\mathit{SatSearch}$ enumerates all lasso-shaped accepting runs $\pref.\rec^\omega$ of $\mathcal{B}$, \ie, accepting runs consisting of a finite prefix $\pref$ and a finite recurring part $\rec$ that is repeated infinitely often. Both $\pref$ and $\rec$ need to end in the same state of $\mathcal{B}$.
Then, the execution effects of $\pref$ and $\pref.\rec$ are computed. $\mathit{SatSearch}$ checks if it is possible to satisfy all predicate constraints induced by $\pref.rec$ under the condition that, for each cell, $\pref$ and $\pref.\rec$ construct equal function terms. For this, it utilizes an SMT solver to check the satisfiability of a quantifier-free first-order logic formula, encoding the consistency requirement, in the theory of equality. If the check succeeds, adding $\rec$ to $\pref$ des not create an inconsistency and hence repeating $\rec$ infinitely often is consistent. Therefore, there exists an execution for $\pref.\rec^\omega$ and thus $\varphi$ is finitary satisfiable in $\tu$ by~\Cref{lem:canonical_model_tu}.

$\mathit{UnsatSearch}$ computes the execution effect of finite sub\-words of runs of~$\mathcal{B}$ and checks whether they are consistent. If a subword is inconsistent, then every run that contains this subword is inconsistent. Hence, there do not exist executions for these runs and therefore they are removed from the set of candidate runs.
If there is no accepting candidate run left, then $\mathcal{B}$ has an empty symbolic language and thus, by \Cref{thm:tsl_to_bsa}, $\varphi$ is unsatisfiable in~$\tu$.

\begin{example}
Consider the finitary BSAs $\mathcal{B}_1$ and $\mathcal{B}_2$ from \Cref{fig:bsa_example_1,fig:bsa_example_2} as well as their respective TSL formulas $\varphi_1 := \Globally \update{\cell{x}}{f(\cell{x})} \land \Globally\Eventually (p(\cell{x}) \land \Next \neg p(\cell{x}))$ and $\varphi_2 := (\Globally (\update{\cell{x}}{f(\cell{x})} \land p(\cell{x})) \land \Eventually \neg p(f(\cell{x}))$.
If we execute \Cref{alg:TSL_TU_SAT_idea} on~$\varphi_1$, $\mathit{SatSearch}$ considers the accepting lasso $q_0 \to q_1 \to q_0$ in $\mathcal{B}_1$ at some point. 
Then, $\pref = \varepsilon$ and $\rec = (q_0,g_1,u,q_1) (q_1,g_2,u,q_0)$. 
Note that $\pref.\rec$ is the finite prefix $c_p$ of a run of $\mathcal{B}_1$ from \Cref{ex:run}. 
Thus, $\executionEffect{\pref.\rec}$ is given by $(\lambda \cell{c}.f(f(\cell{x})), \{ (p(\cell{x}), \true), (p(f(\cell{x})),\false) \})$. 
Since $\executionEffect{\pref} = (\lambda \cell{c}.c, \emptyset)$ holds, $\mathit{SatSearch}$ generates the query $p(\cell{x}) \land \lnot p(f(\cell{x})) \land \cell{x} = f(f(\cell{x}))$ which is satisfiable in $\te$. 
Hence, we can repeat the lasso $q_0 \to q_1 \to q_0$ infinitely often without getting any inconsistent constraints and thus $\varphi_1$ is satisfiable. 

If we execute \Cref{alg:TSL_TU_SAT_idea} on $\varphi_2$, $\mathit{UnsatSearch}$ checks at some point wether in $\mathcal{B}_2$ the transition sequence $q_0 \to q_1$ followed by the upper self-loop is consistent. This is not the case as it requires $p(f(\cell{x}))$ to be true (first transition) and false (second transition): We have $\constraintTrace_1 = \{ (p(\cell{x}),\true),(p(f(\cell{x})),\false) \}$ and $\constraintTrace_2 = \constraintTrace_1 \cup \{ (p(f(\cell{x})),\true),(p(f(f(\cell{x}))),\true) \}$ by definition of the constraint trace.
$\mathit{UnsatSearch}$ also checks the transition sequence $q_0 \to q_1$ followed by the lower self-loop which is also inconsistent. Hence, there is no consistent transition after $q_0 \to q_1$ and thus there is no valid accepting run. Hence, $\varphi_2$ is unsatisfiable.
\end{example}

Note that the presentation of \Cref{alg:TSL_TU_SAT_idea} omits implementation details such as the enumeration of accepting loops and the implementation of the infinite set~$\compBSA$. A more detailed description addressing these issues is given in \Cref{app:fulldescription}.

\Cref{alg:TSL_TU_SAT_idea} is correct. 
Intuitively, it terminates with \sat~if the constraint trace $\constraintTrace$ of the unique computation $\computation$ of $\pref.\rec^\omega$ is consistent. Hence,~$\constraintTrace$ defines an assignment~$\assign$ such that $\symbolicExecution$ is an execution of $\pref.\rec^\omega$, implying satisfiability of $\varphi$ in $\tu$.
If the algorithm terminates with \unsat, then all accepting runs of the BSA are inconsistent and thus no finitary execution satisfying~$\varphi$ exists.
For the proof, we refer the reader to \Cref{app:partial_decider}.

\begin{theorem}[Correctness of \Cref{alg:TSL_TU_SAT_idea}] \label{thm:correctness_algorithm}
	Let $\varphi$ be a TSL formula.
	If \Cref{alg:TSL_TU_SAT_idea} terminates on 	$\varphi$ with~\sat, then there exists an execution $\symbolicExecution$ such that both $\computation \TSLSatisfaction{\mathcal{V}}{\assign} \varphi$ and $\finitary{\varphi}{\computation}$ hold.
	If \Cref{alg:TSL_TU_SAT_idea} terminates with \unsat, then for all executions $\symbolicExecution$ with $\finitary{\varphi}{ctr}$, $\computation \not\TSLSatisfaction{\mathcal{V}}{\assign} \varphi$ holds.
\end{theorem}

\section{Undecidability of TSL modulo $\tu$ Satisfiability}\label{sec:undecidability}

The algorithm for TSL satisfiability checking in $\tu$ presented in the previous section does not necessarily terminate.
In this section, we show that no complete algorithm exists: The satisfiability of a TSL formula in the theory of uninterpreted functions $\tu$ (TSL-$\tu$-SAT) is neither semi-decidable nor co-semi-decidable:

\begin{theorem}[Undecidability of TSL-$\tu$-SAT]\label{thm:undecidability_tsl_modulo_tu}
	The satisfiability (validity) problem of TSL in $\tu$ is neither semi-decidable nor co-semi-decidable.
\end{theorem}

The main intuition behind the undecidability result is that we can encode numbers with TSL in the theory of uninterpreted functions. That is, we are able to encode incrementation, resetting some variable to zero, and equality. We only give the encoding here, for the proof of its correctness we refer to \Cref{app:undecidability}.

\begin{lemma}\label{lem:tu_numbers}
	Let $\encInc$ be a unary function, let $\encEq$ be a binary predicate, and let $\encRes$ be a constant. Let $\encInc^x(\encRes)$ correspond to applying $\encInc$ $x$-times to $\encRes$. There exists a TSL formula $\varphi_\mathit{num}$ such that every execution entailing $\varphi_\mathit{num}$ requires from its models that for all $a,b \in \mathbb{N}_0$, $a = b$ holds if, and only if, $\encInc^a(\encRes) \encEq \encInc^b(\encRes)$ holds.
\end{lemma}
\begin{proofsketch}
	We construct $\varphi_\mathit{num} = \varphi_1 \land \varphi_2$ as follows: The first conjunct is defined by $\varphi_1 := \update{e}{\encRes} \land \Next \Globally\left(\update{e}{\encInc(e)} \land e \encEq e\right)$.	
	Let
	\begin{align*}
		\varphi_\mathit{eq} &:= \left( x \encEq b \right) \rightarrow \left( \update{x}{z} \land \update{b}{\encInc(b)} \land \neg\left( b \encEq \encInc(b) \right) \land \neg\left( \encInc(b) \encEq b \right) \right)\\
		\varphi_\mathit{neq} &:= \neg\left( x \encEq b \right) \rightarrow \left( \update{x}{\encInc(x)} \land \update{b}{b} \land \neg\left( x \encEq \encInc(b) \right) \land \neg\left( \encInc(b) \encEq x \right) \right). 
	\end{align*}
	Then, $\varphi_2$ is defined by $\varphi_2 := \update{x}{z} \land \update{b}{z} \land \Next \Globally (\varphi_\mathit{eq} \land \varphi_\mathit{neq})$.
\end{proofsketch}

Intuitively, $\encInc$ corresponds to incrementation, $\encRes$ to resetting a variable to zero, and $\encEq$ to equality: 
$\varphi_1$ ensures that $\encInc^n(\encRes) \encEq \encInc^n(\encRes)$ holds for all $n \in \mathbb{N}_0$. In contrast,~$\varphi_2$ ensures that if $a \neq b$ holds, then $\neg (\encInc^a(\encRes) \encEq \encInc^b(\encRes))$: Starting with $x = b = z$, $\varphi_1$ ensures that $x \encEq b$ holds initially. Then, $\varphi_\mathit{eq}$ resets $x$ to~$z$ and ``increments'' $b$, while ensuring that $\neg(f^k(z) \encEq f^{k+1}(z))$ holds, where $b = f^k(z)$. Then, $\neg(x \encEq b)$ holds and thus $\varphi_\mathit{neq}$ ``increments'' $x$ until it reaches $b = f^{k+1}(z)$, while ensuring that $\neg(f^{k+1}(z) \encEq f^\ell(z))$ holds for all $\ell < k+1$.

Using this encoding in TSL modulo $\tu$, we can construct a TSL formula $\varphi_\mathcal{G}$ for every \goto-program $\mathcal{G}$ such that $\varphi_\mathcal{G} \land \varphi_\mathit{num}$ is satisfiable in $\tu$ if, and only if,~$\mathcal{G}$ terminates on every input. 
Intuitively, $\varphi_\mathcal{G}$ ``simulates''~$\mathcal{G}$ on different inputs by starting with input zero and incrementing the input if the halting location was reached. The temporal operators of TSL allow for requiring that~$\mathcal{G}$ terminates infinitely often. 
The construction of $\varphi_\mathcal{G}$ is given in \Cref{app:undecidability}. 
Since the universal halting problem for \goto~programs is neither semi-decidable nor co-semi-decidable, the same undecidability result follows for the satisfiability of a TSL formula modulo~$\tu$, proving \Cref{thm:undecidability_tsl_modulo_tu}.

Since the theory of Presburger arithmetic $\tp$ allows for incrementation, resetting a variable to zero, and equality, we can reuse the TSL formula $\varphi_\mathcal{G}$ from above to reduce the universal halting problem for \goto~programs to TSL satisfiability modulo $\tp$ (TSL-$\tp$-SAT), proving undecidability of TSL-$\tp$-SAT. Note that this result holds for other theories that can express incrementation, reset, and equality, for instance Peano Arithmetic, as well.

\begin{theorem}[Undecidability of TSL-$\tp$-SAT]\label{thm:undecidability_tsl_modulo_tp}
	The satisfiability (validity) problem of TSL in $\tp$ is neither semi-decidable nor co-semi-decidable.
\end{theorem}

Furthermore, equality allows for encoding incrementation and resetting a variable to zero. Hence, similarly to $\tu$, there exists a TSL formula $\varphi_\mathit{enc}$ that, if entailed, enforces a binary function and a constant to behave as incrementation and a reset, respectively. The construction of $\varphi_\mathit{enc}$ is given in \Cref{app:undecidability}. Thus, the TSL formula $\varphi_\mathcal{G}$ constructed as above for a \goto~program $\mathcal{G}$ ensures that $\varphi_\mathcal{G} \land \varphi_\mathit{enc}$ is satisfiable in the theory of equality~$\te$ if, and only if, $\mathcal{G}$ terminates on every input. Hence, undecidability of TSL-$\te$-SAT follows:

\begin{theorem}[Undecidability of TSL-$\te$-SAT]\label{thm:undecidability_tsl_modulo_te}
	The satisfiability (validity) problem of TSL in $\te$ is neither semi-decidable nor co-semi-decidable.
\end{theorem}

\section{(Semi-)Decidable Fragments}

In \Cref{sec:undecidability}, we showed that TSL satisfiability is undecidable in~$\tu$.
In this section, however, we identify fragments of TSL on which \Cref{alg:TSL_TU_SAT_idea} terminates for certain inputs. In fact, we present one fragment for which TSL-$\tu$-SAT is decidable and two fragments for which TSL-$\tu$-SAT is semi-decidable.

First, we consider the TSL reachability fragment, \ie, the fragment of TSL that only permits the next operator and the eventually operator as temporal operators.
In our applications, this fragment corresponds to finding counterexamples to safety properties. For satisfiable reachability formulas, \Cref{alg:TSL_TU_SAT_idea} terminates. 
The main idea behind the termination is that the BSA of a reachability formula has an accepting lasso-shaped run and since $\varphi$ is satisfiable, this run is consistent.
For the proof, we refer to \Cref{app:fragments}.

\begin{lemma}\label{lem:semi-decidability_reachability}
    Let $\varphi$ be a TSL formula in the reachability fragment.
    If $\varphi$ is finitary satisfiable in $\tu$, then \Cref{alg:TSL_TU_SAT_idea} terminates~on~$\varphi$.
\end{lemma}

Restricting the reachability fragment further, we consider TSL formulas with updates, predicates, logical operators, next operators, and at most one top-level eventually operator. Such formulas are either completely time-bounded or they are of the form $\varphi = \Eventually \varphi'$, where~$\varphi'$ is time-bounded.
In the dual validity problem, such formulas correspond to invariants on a fixed time window, a useful property for many applications. \Cref{alg:TSL_TU_SAT_idea} is guaranteed to terminate for satisfiable and unsatisfiable formulas of the above form if a \emph{suitable} BSA is constructed.
Such a suitable BSA has a single accepting state $q$ indicating that the time-bounded part has been satisfied. Intuitively, a suitable BSA ensures that all runs reaching $q$ are accepting and that only finitely many transition sequences lead to $q$.
Then, if~$\varphi$ is unsatisfiable, \Cref{alg:TSL_TU_SAT_idea} is able to exclude all transition sequences leading to $q$ and thus to terminate. A BSA with infinitely many transition sequences leading to $q$, in contrast, may cause the algorithm to diverge as it may consider infinitely many consistent subsequences before finding the inconsistent one yielding the exclusion of the sequences leading to $q$.
A suitable BSA exists for every TSL formula in the considered fragment.
For the proof, including a more detailed description of suitable BSAs, we refer to \Cref{app:fragments}.

\begin{lemma}\label{lem:TSL-TU-SAT_termination_top-level_eventually}
    Let $\varphi$ be a TSL formula with only logical operators, predicates, updates, next operators, and at most one top-level eventually operator.
    \Cref{alg:TSL_TU_SAT_idea} terminates on $\varphi$ if it picks a suitable respective BSA.\end{lemma}

Note that \Cref{alg:TSL_TU_SAT_idea} is only a formal decider for this fragment if we ensure that a suitable BSA is always generated. In practice, we experienced that this is usually the case even without posing restrictions on the BSA construction.
Lastly, we consider a fragment of TSL that does not restrict the temporal structure of the formula but the number of used cells. For TSL formulas with a single cell, \Cref{alg:TSL_TU_SAT_idea} always terminates on satisfiable inputs:

\begin{lemma}\label{lem:TSL-TU-SAT_termination_one_cell}
	Let $\varphi$ be a TSL formula such that $|\cells| = 1$. If $\varphi$ is finitary satisfiable in the theory of uninterpreted functions, then \Cref{alg:TSL_TU_SAT_idea} terminates on $\varphi$.	
\end{lemma}

Intuitively, restricting the TSL formula to use only a single cell prevents us from simulating arbitrary computations and thus from reducing from the universal halting problem of \goto~programs as in the general undecidability proof.
The formal proof, given in \Cref{app:fragments}, however, is unrelated to the above intuition.
%
Combining the three observations, we obtain the following (semi-)decidability results for the satisfiability of fragments of TSL modulo $\tu$:

\begin{theorem}
The satisfiability problem of TSL formulas in $\tu$ is (1) semi-decidable for the reachability fragment of TSL, (2) decidable for formulas consisting of only logical operators, predicates, updates, next operators, and at most one top-level eventually operator, and (3) semi-decidable for formulas with one~cell.
\end{theorem}

\section{Evaluation}

We implemented the algorithm for checking TSL modulo $\tu$ satisfiability\footnote{\url{https://github.com/reactive-systems/tsl-satisfiability-modulo-theories}}.
We used \emph{TSL~tools}\footnote{\url{https://github.com/reactive-systems/tsltools}} to handle TSL, \emph{spot}~\cite{DBLP:conf/atva/Duret-LutzLFMRX16} to transform the approximated LTL formulas into NBAs, \emph{SyFCo}~\cite{JacobsFS16} for LTL transformations, and \emph{z3}~\cite{MouraB08} to solve SMT queries.
The implementation follows the extended algorithm description in \Cref{app:fulldescription}.
Since in some cases the default optimizations of \emph{spot} produce a large overhead in computation time, we first execute it with these and if this does not terminate within 20s, we execute it without optimizations.
We evaluated the implementation on three benchmark classes and a machine with an AMD Ryzen~7 processor, using a virtual machine with two logical cores and 6 GB of~RAM.


\begin{figure}[t]
\centering
\begin{minipage}[b]{0.48\textwidth}
\centering
\begin{tikzpicture}
	\path[use as bounding box] (-1.15,-0.9) rectangle (4.3,2.8);

    \begin{axis}[ymode=log, xmin =0, xlabel=Scaling Factor $n$, ylabel=Time(ms), width=\textwidth, height=4.205cm, ylabel near ticks]
        \addplot+[black!40!red, mark options={fill=black!40!red,scale=0.7},semithick] table [x=ID, y=Time, col sep=comma] {satscale.csv};
        \addplot+[black!70!blue, mark options={fill=black!70!blue,scale=0.7},semithick] table [x=ID, y=Time, col sep=comma] {unsatscale.csv};
        \addplot [red, thick] coordinates {(0, 60000) (15, 60000)} node[below] at (axis cs:3, 60000){TO (60s)};
        \addlegendentry{$\varphi_{sat}$};
        \addlegendentry{$\varphi_{unsat}$};
    \end{axis}
\end{tikzpicture}
\captionof{figure}{Execution times in milliseconds of the scalability benchmark series.}\label{fig:scalability}
\end{minipage}
\hfill
\begin{minipage}[b]{0.48\textwidth}
\centering
\begin{tikzpicture}
	\path[use as bounding box] (-1.15,-0.9) rectangle (4.3,2.8);
	
    \begin{axis}[xmin=0, ymode=log, xlabel=Accumulated Instances\vphantom{g}, ylabel=Time(ms), width=\textwidth, height=4.2cm, ylabel near ticks]
        \addplot+[black!70!blue,mark=.,semithick] table [x=ID, y=Time, col sep=comma] {fuzzed.csv};
        \addplot [red, thick] coordinates {(0, 30000) (570, 30000)} node[below] at (axis cs:120, 30000){TO (30s)};
    \end{axis}
\end{tikzpicture}
\captionof{figure}{Execution times in milliseconds of the random benchmark series.}\label{fig:fuzzed}
\end{minipage}
\end{figure}

\paragraph*{Scalability Benchmark Series.} 
We test the scalability of the algorithm with parameterized decidable benchmarks. The timeout is one minute.
Note that \emph{spot} can always perform its optimizations.
The satisfiable benchmarks are defined~by
$\varphi_{sat}(n) := (\LTLglobally \update{\cell{x}}{f(\cell{x})}) \land (\LTLfinally \neg p(\cell{x})) \land \left(\bigwedge_{i=0}^n p(f^i(\cell{x}))\right)$.
The parameter~$n$~corresponds to the number of updates that have to be performed to find a satisfiable lasso.
By \Cref{lem:TSL-TU-SAT_termination_one_cell}, the algorithm always terminates.
The TSL formula
$\varphi_{unsat}(n) := (\LTLglobally (q(\cell{x}) \leftrightarrow \neg q(f^n(\cell{x})))) \land (\LTLglobally \update{\cell{x}}{f(\cell{x})}) \land \LTLfinally (q(\cell{x}) \land \LTLnext^n q(\cell{x}))$
defines the unsatisfiable benchmarks.
The parameter $n$ corresponds to the ``distance'' in time and number of updates of the conflict causing unsatisfiability. The algorithm always terminates. 
The results are shown in \Cref{fig:scalability}. The algorithm scales particularly well for the satisfiable formulas. However, the experiments indicate an exponential complexity of the algorithm for the unsatisfiable formulas.

\paragraph*{Random Benchmark Series.} 
We implemented a random TSL formula generator that uses \emph{spot}'s \texttt{ltlrand} to generate random LTL formulas and then substitutes the atomic propositions with random updates and predicates.
The generated TSL formulas have one to three cells, one to three different updates and one to three different predicates. 
For the LTL formulas generated by \texttt{ltlrand} we use tree sizes from 5 to 95 in steps of five.
For each of the tree sizes, we generate 30 formulas; 570 in total. The execution times are shown in \Cref{fig:fuzzed}. On many formulas, the algorithm terminates within one second.
The implementation returns \sat~for 513 of the 570 formulas. 
It times out after 30s on 29 formulas. However, the timeouts already occur in the automaton construction, both with and without \emph{spot}'s optimizations. 
Only 28 formulas are unsatisfiable. For 25 of these unsatisfiable formulas, the intermediate LTL approximation formula is already unsatisfiable, \ie, only for three formulas there is some conflict due to updates and predicate evaluation.

\paragraph*{Applications Benchmark Series.}
These benchmarks correspond to checking consistency of a specification and validating assumptions of a system. Hence, they illustrate how satisfiability results can aid the system design.
The results are presented in \Cref{tab:evaluation_application}.
We introduce two of the benchmarks in more detail here. The other, slightly larger, ones, including different kinds of arbiters, a scheduler, and modules of the Syntroids~\cite{GeierH0F19} arcade game, are described in \Cref{app:other_benchmarks}.

\begin{table}[t]
	\caption{Execution times in seconds of the application benchmark series.}\label{tab:evaluation_application}
	\centering
	\begin{minipage}[b]{0.52\textwidth}
	\centering
	\begin{tabular}[t]{p{3.5cm}|>{\centering}p{1.6cm}c}
		Benchmark & Result & Time \\\hline\hline
		Chain&\sat&7.06\\
		Filter&\unsat&0.33\\
		Gamemodechooser Ass.&\unsat&35.55\\
	    Holding Arbiter&\sat&11.75\\
	    Small Holding Arbiter&\sat&36.69\\
	    P.\ T.\ Arbiter&\unsat&56.03\\
	    Approx.\ P.\ T.\ Arbiter&\unsat&940.03
	\end{tabular}
	\end{minipage}%
	\hfill
	\begin{minipage}[b]{0.45\textwidth}
	\centering
	\begin{tabular}[t]{p{2.6cm}|>{\centering}p{1.6cm}c}
		Benchmark & Result & Time \\\hline\hline
		Inductive Ass.&\unsat&0.25\\
	    One Of Two&\unsat&1.20\\
	    One Of Three&\unsat&4.25\\
	    Injector&\unsat&1.52\\
	    Invariant Holding&\unsat&2.33\\
	    Scheduler&\unsat&3.87
	\end{tabular}
	\end{minipage}
\end{table}

The \emph{Chain} benchmark considers a compound system of two chained modules $m_1$ and $m_2$ that receive an input value, store it, and forward it to the next system: $\varphi_i := \LTLglobally \LTLfinally (\update{\cell{mem}_i}{\cell{in}_i} \land \LTLnext \update{\cell{in}_{i + 1}}{\cell{mem}_i})$ for $i \in \{1,2\}$.
To simulate the input of the first module, we use an update with an uninterpreted function: $\spec{inp} := \LTLglobally \update{\cell{in}_1}{f(\cell{in}_1)}$.
We require that if some property $p$ holds on $m_1$'s input, $p$ also needs to hold hold on $m_2$'s output: $\spec{spec} := \LTLglobally (p(\cell{in}_1) \to \LTLfinally p(\cell{in}_3))$.
Our algorithm determines within~8s that $(\spec{inp} \land \varphi_1 \land \varphi_2) \land \lnot \spec{spec}$ is satisfiable, detecting an inconsistency:
If $m_1$ stores some value on which $p$ holds, it may overwrite it before $m_2$ copies it, preventing the value to reach $m_2$'s output.

The \emph{Filter} benchmark studies a system that ``passes through'' an input value to a cell if it fulfills a certain property $p$ and holds the previous value otherwise:
$\spec{filter} := \update{\cell{out}}{\mathit{d}()} \land \LTLnext\LTLglobally ((p(\cell{in}) \to \update{\cell{out}}{\cell{in}}) \land (\lnot p(\cell{in}) \to \update{\cell{out}}{\cell{out}}))$, where $d$ is a constant representing an initial default value. The default value fulfills $p$, \ie, $\spec{fact} := \LTLglobally \predconst{d}$.
As for the chain, $\spec{inp} := \LTLglobally \update{\cell{in}}{f(\cell{in})}$ simulates the input.
The filter is valid if $p$ holds on all outputs after the initialization: $\spec{spec} := \LTLnext\LTLglobally p(\cell{out})$. 
Within 400ms, the algorithm confirms that $(\spec{inp} \land \spec{fact} \land \spec{filter}) \land \lnot \spec{spec}$ is unsatisfiable, validating the filter.

\section{Related Work}

Linear-time temporal logic (LTL)~\cite{Pnueli77} is one of the most popular specification languages for reactive systems. It is based on an underlying assertion logic, such as propositional logic, which is extended with temporal modalities.
Satisfiability of propositional LTL has long known to be decidable~\cite{SistlaC85} and there are efficient tools for LTL satisfiability checking~\cite{RozierV07,LiZPVH13}.

While propositional LTL is very common, especially in hardware verification, LTL with richer assertion logics, such as first-order logic and various theories, have long been used in verification (cf.~\cite{Manna+Pnueli/81/Verification}).
Temporal Stream Logic (TSL)~\cite{FinkbeinerKPS19} was introduced as a new temporal logic for reactive synthesis. In the original TSL semantics, all functions and predicates are uninterpreted. TSL synthesis is undecidable in general, even without inputs or equality, but can be under-approximated by the decidable LTL synthesis problem~\cite{FinkbeinerKPS19}. TSL has been used to specify and synthesize an arcade game realized on an FPGA~\cite{GeierH0F19}.

Constraint LTL (CLTL)~\cite{DemriD07} extends LTL with the possibility of expressing constraints between variables at bounded distance. A constraint system $\mathcal{D}$ consists of a concrete domain and an interpretation of relations on the domain. In Constraint LTL over $\mathcal{D}$ (CLTL($\mathcal{D}$)), one can relate variables with relations defined in $\mathcal{D}$.  
Similar to updates in TSL, CLTL can specify assignment-like statements by utilizing the equality relation.
Like for all constraints allowing for a counting mechanism, LTL with Presburger constraints, \ie, CLTL($\mathbb{Z},=,+$), is undecidable~\cite{DemriD07}. However, there exist decidable fragments such as LTL with finite constraint systems~\cite{Demri06b} and LTL with integer periodicity constraints~\cite{Demri06a}.
Permitting constraints between variables at an unbounded distance leads to undecidability even for constraint systems that only allow equality checks on natural numbers. Restricting such systems to a finite number of constraints yields decidability again~\cite{DemriLN05}.
In TSL modulo theories, a theory is given from which a model can be chosen. In CLTL, in contrast, the concrete model is fixed. Therefore, TSL modulo theories cannot be encoded into CLTL in general.

LTL has been extended with the \emph{freeze operator}~\cite{DemriL06,DemriDG07}, allowing for storing an input in a register.
 Then, the stored value can be compared with a current value for equality. Freeze LTL with two registers is undecidable~\cite{LisitsaP05,FreezeLTL} . For flat formulas, \ie, formulas where the possible occurrences of the freeze operator are restricted, decidability is regained~\cite{FreezeLTL}.
Similar to the freeze operator, updates in TSL allow for storing values in cells and in TSL modulo the theory of equality the equality check can be performed. 
In TSL, we can perform computations on the stored values which is not possible in freeze LTL. Hence, freeze LTL can be seen as a special case of TSL.
Constraint LTL has been augmented with the freeze operator as well~\cite{FreezeLTL}. For an infinite domain equipped with the equality relation, it is undecidable. For flat formulas, decidability is regained~\cite{FreezeLTL}.

The temporal logic of actions (TLA)~\cite{Lamport94} is designed to model computer systems. States are assignments of values to variables and actions relate states. Actions can, similar to updates in TSL, describe assignments of variables. A TLA formula may contain state functions and predicates. Actions and state functions are combined with the temporal operators $\Globally$ and $\Eventually$. In contrast to TSL, $\Next$ and $\Until$ are not permitted.
The validity problem for the propositional fragment of TLA, \ie, with uninterpreted functions and predicates, is PSPACE complete~\cite{Ramakrishna95}.

Similar to temporal logics, dynamic logic~\cite{Pratt76,HarelMP77} is an extension of modal logic to reason about computer programs. Dynamic logic allows for stating that after action $a$, it is necessarily the case that $p$ holds, or it is possible that $p$ holds. Compound actions can be build up from smaller actions. In propositional dynamic logic (PDL)~\cite{FischerL79}, data is omitted, \ie, its terms are actions and propositions. PDL satisfiability is decidable in EXPTIME~\cite{Pratt78}.
First-order dynamic logic (FODL)~\cite{Harel79} allows for including data: First-order quantification over a first-order structure, the so-called domain of computation, is allowed.
Dynamic logic does not contain temporal operators such as $\Globally$ or $\Eventually$.
Since we consider reactive systems, \ie, systems that continually interact with their environment, temporal logics are better suited than dynamic logics for our setting.

\emph{Symbolic automata}~(see \eg\cite{DAntoniV17,DAntoniV21}) and \emph{register automata}~\cite{RegisterAutomata} are extensions of finite automata that are capable of handling large or infinite alphabets. 
Register automata have additionally been considered over infinite words in some works (see \eg\cite{DemriL06,KhalimovMB18,ExibardFR21}).
Similar to BSAs, transitions of symbolic automata are labeled with predicates over a domain of alphabet symbols.
Register automata are equipped with a finite amount of registers that, similar to cells in BSAs, can store input values.
\emph{Symbolic register automata}~(SRAs)~\cite{DAntoniFS019} combine the features of both automata models.
BSAs have the additional ability to modify the stored values and thus to perform actual computations on them. Moreover, they read infinite instead of finite words.
Thus, SRAs can be seen as a special case of BSAs.

More recently, the verification of uninterpreted programs has been investigated~\cite{MathurMV19}. Uninterpreted programs are similar to \while-programs with equality and uninterpreted functions and predicates. They are annotated with assumptions. The verification of uninterpreted programs is undecidable in general; for the subclass of coherent uninterpreted programs, however, it is decidable~\cite{MathurMV19}.
The verification problem has been extended with theories, \ie, with axioms over the functions and predicates~\cite{MathurM020}. Adding axioms to coherent uninterpreted programs preserves decidability for some axioms, \eg, idempotence, while it yields undecidability for others, \eg, associativity. The synthesis problem for uninterpreted programs is undecidable in general, but decidable for coherent ones~\cite{Krogmeier0MM020}.

\section{Conclusion}

We have extended Temporal Stream Logic (TSL) with first-order theories and formalized the satisfiability and validity of a TSL formula in a theory.
While we show that TSL satisfiability is neither semi-decidable nor co-semi-decidable in the theory of uninterpreted functions~$\tu$, the theory of equality~$\te$, and Presburger arithmetic~$\tp$, we identify three fragments for which satisfiability in $\tu$ is (semi-)decidable: For reachability formulas as well as for formulas with a single cell, TSL satisfiability in $\tu$ is semi-decidable. For slightly more restricted reachability formulas, it is decidable.
Moreover, we have presented an algorithm for checking the satisfiability of a TSL formula in the theory of uninterpreted functions that is based on Büchi stream automata, an automaton representation of TSL formulas introduced in this paper.
Satisfiability checking has various applications in the specification and analysis of reactive systems such as identifying inconsistent requirements during the design process. We have implemented the algorithm and evaluated it on three different benchmark series, including consistency checks and assumption validations: 
The algorithm~terminates on many randomly generated formulas within one second and scales particularly well for satisfiable formulas. Moreover, it is able to prove or disprove consistency of realistic benchmarks and to validate or invalidate their assumptions.

\bibliographystyle{splncs04}
\bibliography{bib}


\appendix
\section{Preliminaries: First-Order Logic}\label{app:fol}

In our context, first-order logic (FOL) consists of classical logical operators, existential and universal quantification of variables, and predicate terms. 
This allows, on the one hand, to express statements using functions and predicates and, on the other hand, to make statements on functions and predicates. 

Let $\Sigma_F$, $\Sigma_P$, and $V$ be sets of function symbols, and predicate symbols, and variables, respectively. Then, FOL is defined by 
\[ \normalfont{FOL} \ni \varphi, \psi ::= \top \mid \neg \varphi \mid \varphi \land \psi \mid \exists x.\psi \mid \forall x.\psi \mid \predTerm \] 
where $x \in V$. A formula $\varphi$ built according to this grammar is called \emph{first-order logic formula over $\Sigma_F$, $\Sigma_P$, and $V$}. Variables that are not bound by a quantifier are called \emph{free} and the free variables of a FOL formula $\varphi$ are denoted by $\free{\varphi}$. A FOL formula $\varphi$ is \emph{closed} if $\free{\varphi} = \emptyset$. Given a domain $\mathcal{V}$ and an assignment function $\assign : V \cup \Sigma_F \to \mathcal{V} \cup \mathcal{F}$, entailment of a FOL formula is defined by:
\begin{alignat*}{3}
    &\mathcal{V},\assign \vdash \top ~ & & \\
    &\mathcal{V},\assign \vdash \neg \varphi ~ &:\Leftrightarrow & ~ \mathcal{V}, \assign \not\vdash \varphi\\
	&\mathcal{V},\assign \vdash \varphi \land \psi ~ &:\Leftrightarrow & ~ \mathcal{V}, \assign \vdash \varphi ~ \land ~ \mathcal{V}, \assign \vdash \psi\\
	&\mathcal{V},\assign \vdash \exists x. \varphi ~ &:\Leftrightarrow & ~ \exists v \in X.~ \mathcal{V}, \assign[x:=v] \vdash \varphi\\
	&\mathcal{V},\assign \vdash \forall x. \varphi ~ &:\Leftrightarrow & ~ \forall v \in X.~ \mathcal{V}, \assign[x:=v] \vdash \varphi\\
	&\mathcal{V},\assign \vdash \predTerm ~ &:\Leftrightarrow & ~ \recursiveAssign{\predTerm} = \true,
\end{alignat*}
where $\assign[x:=v] := z \mapsto \begin{cases} v &\text{ if } z = x \\ \langle z \rangle &\text{ otherwise} \end{cases}$ which corresponds to overwriting the value of a variable.


\section{TSL modulo $\tu$ Satisfiability Checking}\label{app:partial_decider}


\subsection{Translating a TSL formula into a BSA (Proof of \Cref{thm:tsl_to_bsa})}

In this section, we present a translation from TSL formulas to Büchi stream automata that ensures that (finitary) satisfiability of a TSL formula $\varphi$ in a theory~$T$ directly corresponds to the emptiness of the language of the corresponding (finitary) Büchi stream automaton $\mathcal{B}_\varphi$ with respect to $T$. That is, we prove \Cref{thm:tsl_to_bsa} in the following.

First, we present a construction of the corresponding BSA $\mathcal{B}_\varphi$ from a TSL formula $\varphi$ such that $\mathcal{B}_\varphi$ accepts exactly those executions that satisfy~$\varphi$:

\begin{lemma}\label{lem:tsl_to_bsa_non-fini}
	Given a TSL formula $\varphi$, there exists an equivalent BSA~$\mathcal{B}_\varphi$ such that for all executions $\symbolicExecution$, we have $\symbolicExecution \in \mathcal{L}(\mathcal{B}_\varphi)$ if, and only if, $\computation \TSLSatisfaction{\mathcal{V}}{\assign} \varphi$.
\end{lemma}
\begin{proof}
Let $\varphi$ be a TSL formula over $\Sigma_F$, $\Sigma_P$, and $\cells$. We construct a Büchi Stream automaton~$\mathcal{B}_\varphi$ in three steps: First, we approximate $\varphi$ by an LTL formula $\varphi_{LTL}$. Second, we construct a nondeterministic Büchi automaton $\mathcal{A}_{\varphi_{LTL}}$ with $\mathcal{L}(\varphi_{LTL}) = \mathcal{L}(\mathcal{A}_{\varphi_{LTL}})$. Third, we construct the desired BSA $\mathcal{B}_\varphi$ from $\mathcal{A}$.

\begin{itemize}
	\item[(1)] LTL~\cite{Pnueli77} has the same boolean and temporal operators as LTL. Yet, it is defined over boolean atomic propositions instead of predicates and updates. In $\varphi_{LTL}$, we represent predicates and updates by atomic propositions: We fix the set of atomic propositions for $\varphi_{LTL}$ to be \[AP := \allUpdates{\varphi} \cup \{ \update{\cell{c}}{\fresh} \mid \cell{c} \in \cells \} \cup \allPredicates{\varphi}, \] where $\fresh \not\in \Sigma_F \cup \Sigma_P \cup \cells$. Let $\varphi'_{LTL}$ be a copy of $\varphi$, where the predicates and updates are interpreted as atomic propositions. It remains to ensure that exactly one update is performed for every cell in every time step. Hence, we define $\varphi_{LTL} := \varphi'_{LTL} \land \varphi_\mathit{least} \land \varphi_\mathit{most}$, where \[ \varphi_\mathit{least} := \Globally \bigwedge_{\cell{c} \in \cells} \left( \bigvee_{\update{\cell{c}}{t_f} \in AP} \update{\cell{c}}{t_f} \right) \] \[ \varphi_\mathit{most} :=  \Globally \bigwedge_{\substack{\update{\cell{c}}{t_f},\update{\cell{c}}{t_g} \in AP, \\ t_f \neq t_g}} \neg \left( \update{\cell{c}}{t_f} \land \update{\cell{c}}{t_g} \right) \] Note that a similar approximation is described in \cite{FinkbeinerKPS19}.
	\item[(2)] For $\varphi_{LTL} = \varphi'_{LTL} \land \varphi_\mathit{least} \land \varphi_\mathit{most}$, we construct an equivalent nondeterministic Büchi automaton $\mathcal{A}_{\varphi_{LTL}} = (Q,Q_0,\delta,F)$ with $\mathcal{L}(\varphi_{LTL}) = \mathcal{L}(\mathcal{A}_{\varphi_{LTL}})$~\cite{VardiW94}.
	\item[(3)] From $\mathcal{A}_{\varphi_{LTL}}$, we construct the desired BSA $\mathcal{B}_\varphi = (Q^\mathcal{B},Q^\mathcal{B}_0,F^\mathcal{B},\fresh,\guards,\updateTerms,\delta^\mathcal{B})$ as follows: We keep the structure of $\mathcal{A}_{\varphi_{LTL}}$, \ie, we define $Q^\mathcal{B} := Q$, $Q^\mathcal{B}_0 := Q_0$, and $F^\mathcal{B} := F$. The guards of $\mathcal{B}_\varphi$ are given by the predicates occurring in~$\varphi$, \ie, we have $\guards := \allPredicates{\varphi}$, and the update terms are given by the updates occurring in $\varphi$ as well as the updates of all cells with $\fresh$, \ie, we have $\updateTerms := \allUpdates{\varphi} \cup \{ \update{\cell{c}}{\fresh} \mid \cell{c} \in \cells \}$. The transition relation $\delta^\mathcal{B}$ accounts for having predicates and updates instead of atomic propositions in a BSA: We define $(q,P,U,q') \in \delta^\mathcal{B}$ if, and only if, $(q,\operatorname{ap}(P,U),q')\in\delta$, where \[\operatorname{ap}(P,U) := \{ p \in \allPredicates{\varphi} \mid P(p) = \true \} \cup \{ \update{\cell{c}}{U(\cell{c})} \mid \cell{c} \in \cells \}.\] That is, a predicate on a transition holds if, and only if, the respective atomic proposition is true. An update is executed if, and only if, the atomic proposition encoding it is true. Note that the transition structure of $\mathcal{B}_\varphi$ is a substructure of the one of $\mathcal{A}_{\varphi_{LTL}}$. Hence, a transition in $\mathcal{A}_{\varphi_{LTL}}$ does not necessarily have a counterpart in $\mathcal{B}_\varphi$. In fact, a transition of $\mathcal{A}_{\varphi_{LTL}}$ has no counterpart in $\mathcal{B}_\varphi$ if it allows for more than one update or no update at all for one cell. However, since $\varphi_\mathit{least}$ and $\varphi_\mathit{most}$ enforce exactly one update for each cell in every time step, such a transition inevitably leads to a rejecting run. More precisely: Every run containing such a transition is rejecting. Therefore, such transitions can be omitted in $\mathcal{B}_\varphi$.
\end{itemize}

The correctness of this construction directly follows from the correctness of the LTL to NBA construction since LTL and NBAs have the same control structure as TSL and BSAs. Furthermore, the construction of $\varphi_{LTL}$ ensures that the data manipulation properties, \ie, exactly one update for each cell at every time step, are transferred correctly.
\end{proof}

Second, we extend this result to finitary BSA and finitary satisfiability:

\begin{lemma}\label{lem:tsl_to_bsa_fini}
	Given a TSL formula $\varphi$, there exists an equivalent finitary BSA~$\mathcal{B}_\varphi$ such that for all executions $\symbolicExecution$, we have $\symbolicExecution \in \mathcal{L}(\mathcal{B}_\varphi)$ if, and only if, both $\computation \TSLSatisfaction{\mathcal{V}}{\assign} \varphi$ and $\finitary{\varphi}{ctr}$ hold.
\end{lemma}
\begin{proof}
Let $\varphi$ be a TSL formula over $\Sigma_F$, $\Sigma_P$, and $\cells$. We construct a Büchi Stream automaton $\mathcal{B}_\varphi$ similar to the non-finitary case, \ie, as in the proof of \Cref{lem:tsl_to_bsa_non-fini}. The only difference is that we define the set of atomic propositions for $\varphi_{LTL}$ to be \[AP := \allUpdates{\varphi} \cup \{ \update{\cell{c}}{\cell{c}} \mid \cell{c} \in \cells \} \cup \allPredicates{\varphi}\] to only allow updates occurring in $\varphi$ and self updates. By construction of $\mathcal{B}_\varphi$, an update with $\fresh$ does not appear in any transition. Thus, by definition, $\mathcal{B}_\varphi$ is finitary. Furthermore, by construction of $AP$, transitions of $\mathcal{B}_\varphi$ may only contain updates occurring in $\varphi$ or self updates. Hence, $\mathcal{B}_\varphi$ only accepts finitary executions.
\end{proof}

With these automata constructions, we can now show the claim of \Cref{thm:tsl_to_bsa}: For every TSL formula $\varphi$, there is a (finitary) BSA $\mathcal{B}_\varphi$ such that the (finitary) satisfiability of $\varphi$ in a theory $T$ can be solved by checking the emptiness of the language of $\mathcal{B}_\varphi$ with respect to $T$.

Let $\varphi$ be a TSL formula and let $T$ be a theory. As in the proofs of \Cref{lem:tsl_to_bsa_non-fini,lem:tsl_to_bsa_fini}, we can construct a (finitary) BSA $\mathcal{B}_\varphi$ such that $\mathcal{B}_\varphi$ accepts exactly those (finitary) executions that satisfy $\varphi$. By definition of the language of a BSA with respect to $T$, we have $\theoryLanguage{T}{\mathcal{B}_\varphi} \neq \emptyset$ if, and only if, there exists an execution $\symbolicExecution \in \theoryLanguage{\varphi}{\mathcal{B}}$ such that $(\mathcal{V},\assign) \in \allModels{T}$. By \Cref{lem:tsl_to_bsa_non-fini}, we have $\computation \TSLSatisfaction{\mathcal{V}}{\assign} \varphi$ since $\symbolicExecution \in \mathcal{L}(\mathcal{B}_\varphi)$. \Cref{lem:tsl_to_bsa_fini} adds that $\computation$ is finitary. Thus, by definition of (finitary) TSL satisfiability, $\varphi$ is (finitary) satisfiable in~$T$.


\subsection{Reduction to Finitary Satisfiability (Proof of \Cref{lem:reduction_to_finitary})}

Let $\varphi$ be a TSL formula over $\Sigma_F$, $\Sigma_P$, and $\cells$.
We define the desired TSL formula~$\varphi_\mathit{fin}$ as follows: 
\begin{align*}
    \varphi_\mathit{fin} := 
    & \Globally \; \update{\cell{n}}{\mathit{new}(\cell{n})} \; \land \\
    & \Globally \bigwedge_{\cell{c}\in\allCells{\varphi}} \left(\left( \bigvee_{\update{\cell{c}}{t_f}  \in \allUpdates{\varphi}} \update{\cell{c}}{t_f} \right) \lor \update{\cell{c}}{\mathit{pick}_\cell{c}(\cell{n})} \right)
\end{align*}
where $\allCells{\varphi}$ and $\allUpdates{\varphi}$ denote the sets of all cells and all updates occurring in $\varphi$, respectively.
Note that, in contrast to $\varphi$, $\varphi_\mathit{fin}$ is a TSL formula over function symbols $\Sigma_F \cup \{ \mathit{new} \} \cup \{ \mathit{pick}_\cell{c} : \cell{c} \in \cells \}$, where $\mathit{new}, \mathit{pick}_\cell{c} \not\in \Sigma_F$ for any $\cell{c} \in \cells$, predicate symbols $\Sigma_P$, and cells $\cells \cup \{\cell{n}\}$, where $\cell{n} \not\in \cells$. Intuitively, $\mathit{pick}_\cell{c}$ represents the ``unknown update'' depending on the time step, represented by $n$.

Now, we show that $\varphi_\mathit{fin}$ satisfies the required property, \ie, that $\varphi$ is satisfiable in $\tu$ if, and only if, $\varphi \land \varphi_\mathit{fin}$ is finitary satisfiable in~$\tu$.

First, consider the right-to-left direction. That is, let $\varphi \land \varphi_\mathit{fin}$ be finitary satisfiable in~$\tu$. Then, by definition, there exists an execution $(\computation,\mathcal{V},\assign)$ such that $\computation \TSLSatisfaction{\mathcal{V}}{\assign} \varphi \land \varphi_\mathit{fin}$ holds. Hence, $\computation \TSLSatisfaction{\mathcal{V}}{\assign} \varphi$ follows with the semantics of conjunction. Thus, by \Cref{lem:canonical_model_tu}, $\varphi$ is satisfiable in $\tu$.

Next, consider the left-to-right direction. That is, let $\varphi$ be satisfiable in $\tu$. Without loss of generality, let $\cells = \allCells{\varphi}$\footnote{This can easily be generalized to the case where $\cells \neq \allCells{\varphi}$ by adding self-updates for all cells $\cell{c} \in \cells \setminus \allCells{\varphi}$ to $\varphi$ since $\varphi$ does not restrict them.}. 
Since $\varphi$ is satisfiable in $\tu$ by assumption, there exists an execution $(\computation,\mathcal{V},\assign)$ such that $\computation \TSLSatisfaction{\mathcal{V}}{\assign} \varphi$ holds.
We construct an execution $(\computation_\mathit{fin}, \mathcal{V}_\mathit{fin}, \assign_\mathit{fin})$ for $\varphi_\mathit{fin}$ in the following:
Let $\computation_\mathit{fin}$ be a computation defined by \[ \atPos{\computation_\mathit{fin}}{t}(\cell{c}) := \begin{cases}
    \mathit{new}(\cell{n}) & \text{if } \cell{c} = \cell{n} \\
		\atPos{\computation}{t}(\cell{c}) & \text{if } \cell{c} \neq \cell{n} \text{ and } \update{\cell{c}}{\atPos{\computation}{t}(\cell{c})} \in \allUpdates{\varphi} \cup \update{\cell{c}}{\cell{c}} \\
        \mathit{pick}_\cell{c}(\cell{n}) & \text{otherwise}
	\end{cases} \]
This computation $\computation_\mathit{fin}$ is a copy of $\computation$ that replaces every non-finitary part of $\computation$ by a new term as described by $\varphi_\mathit{fin}$. Next, let $\mathcal{V}_\mathit{fin} := \mathcal{V} \cup \{ \mathit{new}^m(\cell{n}) : m \in \mathbb{N}_0 \}$. Last, we define an assignment function $\assign_\mathit{fin}$ that behaves as $\assign$. 
\[\langle f \rangle_\mathit{fin} := 
    \begin{cases}
        \lambda v.
        \begin{cases}        
            \recursiveAssign{\etaFunction{\computation}{m+1}{\cell{c}}}   
                & \text{if } v = \mathit{new}^m(\cell{n}) \\
            v   & \text{otherwise}
        \end{cases}
                          & \text{if } f = \mathit{pick}_\cell{c} \\
        \lambda v.\mathit{new}(v)   & \text{if } f = \mathit{new} \\
        \cell{n}                    & \text{if } f = \cell{n} \\
        \langle f \rangle           & \text{otherwise} \\
    \end{cases}
\]

We first show that both $\assign_\mathit{fin}$ and $\assign$ indeed behave similar:
\begin{lemma}\label{lem:correct_construction_fin}
	For all predicate terms~$t_p\in\predTerms$ occurring in $\varphi$ and all points in time~$t\in\mathbb{N}_0$, $\recursiveAssignFin{\etaFunction{\computation_\mathit{fin}}{t}{t_p}} = \recursiveAssign{\etaFunction{\computation}{t}{t_p}}$ holds.
\end{lemma}
\begin{proof}
    First observe that for all function terms $t_f \in \funcTerms$ that do not contain $\mathit{new}$, $\mathit{pick}_\cell{c}$, or~$\cell{n}$, $\recursiveAssignFin{t_f} = \recursiveAssign{t_f}$ holds (Observation X). 
    Since $\chi$ only recursively evaluates terms, we obtain that if each recursive call is equal, then the results of $\assign_\mathit{fin}$ and $\assign$ are also equal. 
    In addition, $t_p$ does not contain $\mathit{new}$, $\mathit{pick}_\cell{c}$, or $\cell{n}$.
    Hence, it suffices to show that \[\recursiveAssignFin{\etaFunction{\computation_\mathit{fin}}{t}{\cell{c}}} = \recursiveAssign{\etaFunction{\computation}{t}{\cell{c}}}\] holds for $\cell{c} \in \cells$. 
    We prove this by induction on $t$:
    For $t = 0$, clearly \[\etaFunction{\computation}{t}{\cell{c}} = \cell{c} = \etaFunction{\computation_\mathit{fin}}{t}{\cell{c}}\] holds.
    Now let $t > 0$.  We distinguish two cases.
    \begin{itemize}
    \item  
    If $\atPos{\computation_\mathit{fin}}{t-1}(\cell{c}) = \atPos{\computation}{t-1}(\cell{c})$ holds, then
    \begin{align*}
         & \recursiveAssignFin{\etaFunction{\computation_\mathit{fin}}{t}{\cell{c}}} \\
        =& ~\recursiveAssignFin{\etaFunction{\computation_\mathit{fin}}{t-1}{\atPos{\computation_\mathit{fin}}{t-1}(\cell{c})}} & \text { Definition} \\
        =& ~\recursiveAssignFin{\etaFunction{\computation_\mathit{fin}}{t-1}{\atPos{\computation}{t-1}(\cell{c})}} & \text{ Assumption}\\
        =& ~\recursiveAssign{\etaFunction{\computation}{t-1}{\atPos{\computation}{t-1}(\cell{c})}} & \text{ Observation (X) and IH}\\
        =& ~\recursiveAssign{\etaFunction{\computation}{t}{\cell{c}}} & \text{ Definition} 
    \end{align*}
    \item
        If $\atPos{\computation_\mathit{fin}}{t-1}(\cell{c}) \neq \atPos{\computation}{t-1}(\cell{c})$ holds, then $\atPos{\computation_\mathit{fin}}{t-1}(\cell{c}) = \mathit{pick}_\cell{c}(\cell{n})$ holds as well by construction.
        Hence, $\recursiveAssignFin{\etaFunction{\computation_\mathit{fin}}{t}{\cell{c}}} = \langle \mathit{pick}_\cell{c} \rangle (\etaFunction{\computation_\mathit{fin}}{t-1}{\cell{n}})$. Since, by construction, $\etaFunction{\computation_\mathit{fin}}{t-1}{\cell{n}} = \mathit{new}^{t-1}(\cell{n})$ holds, we have 
        \[
            \langle \mathit{pick}_\cell{c} \rangle (\etaFunction{\computation_\mathit{fin}}{t-1}{\cell{n}})
            = \langle \mathit{pick}_\cell{c} \rangle (\mathit{new}^{t-1}(\cell{n}))
            = \recursiveAssign{\etaFunction{\computation}{t}{\cell{c}}}.\]
    Hence $\recursiveAssignFin{\etaFunction{\computation_\mathit{fin}}{t}{\cell{c}}} = \recursiveAssign{\etaFunction{\computation}{t}{\cell{c}}}$ follows.
    \end{itemize}
\end{proof}

To prove that $\varphi \land \varphi_\mathit{fin}$ is finitary satisfiable in $\tu$, by \Cref{lem:canonical_model_tu}, it suffices to show that $\computation_\mathit{fin} \TSLSatisfaction{\mathcal{V}_\mathit{fin}}{\assign_\mathit{fin}} \varphi \land \varphi_\mathit{fin}$ holds and that $\computation_\mathit{fin}$ is finitary.
We start with the second claim and prove it for all three cases of the definition of $\computation_\mathit{fin}$ separately:
\begin{enumerate}
    \item 
        Let $\atPos{\computation_\mathit{fin}}{t}(\cell{c}) = \mathit{new}(\cell{n})$ for $\cell{c}=\cell{n}$. Clearly, in this case $\computation_\mathit{fin}$ is finitary since we have $\update{\cell{n}}{\mathit{new}(\cell{n})} \in \allUpdates{\varphi_\mathit{fin}}$.
    \item 
        Let $\atPos{\computation_\mathit{fin}}{t}(\cell{c}) = \atPos{\computation}{t}(\cell{c})$ for $c \in \cells$. Clearly, in this case $\computation_\mathit{fin}$ is finitary since it requires $\update{\cell{c}}{\atPos{\computation}{t}(\cell{c})} \in \allUpdates{\varphi} \cup \update{\cell{c}}{\cell{c}}$.
    \item 
        Let $\atPos{\computation_\mathit{fin}}{t}(\cell{c}) = \mathit{pick}_\cell{c}(\cell{n})$. Clearly, also in this last case $\computation_\mathit{fin}$ is finitary since we have $\update{\cell{c}}{\mathit{pick}_\cell{c}(\cell{n})} \in \allUpdates{\varphi_\mathit{fin}}$.
\end{enumerate}	

Second, we show that $\computation_\mathit{fin} \TSLSatisfaction{\mathcal{V}_\mathit{fin}}{\assign_\mathit{fin}} \varphi \land \varphi_\mathit{fin}$ holds.
By the structure of $\varphi_\mathit{fin}$, we have $\computation_\mathit{fin} \TSLSatisfaction{\mathcal{V}_\mathit{fin}}{\assign_\mathit{fin}} \varphi_\mathit{fin}$ if $\atPos{\computation_\mathit{fin}}{t}(\cell{n}) = \update{\cell{n}}{\mathit{new}(\cell{n})}$ and for $\cell{c} \in \cells$ either $\atPos{\computation_\mathit{fin}}{t}(\cell{c}) = \mathit{pick}_\cell{c}(\cell{n})$ holds or, if there is an update $\update{\cell{c}}{t_f}$ occurring in $\varphi$, $\atPos{\computation_\mathit{fin}}{t}(\cell{c}) = t_f$ holds. 
By construction of $\computation_\mathit{fin}$, this is clearly the case. 
Thus, it remains to show that $\computation_\mathit{fin} \TSLSatisfaction{\mathcal{V}_\mathit{fin}}{\assign_\mathit{fin}} \varphi$ holds. Instead of proving it directly, we show the more general statement $ \forall t.~ (\computation, t \TSLSatisfaction{\mathcal{V}}{\assign} \varphi) \rightarrow (\computation_\mathit{fin}, t \TSLSatisfaction{\mathcal{V}_\mathit{fin}}{\assign_\mathit{fin}} \varphi) $.The original claim then follows by instantiation of $t$ with zero and the fact that $\computation \TSLSatisfaction{\mathcal{V}}{\assign} \varphi$ holds.
We prove the generalized statement by induction on $\varphi$:
\begin{itemize}
	\item Let $\varphi = \true$. By the semantics of TSL, $\computation_\mathit{fin} \TSLSatisfaction{\mathcal{V}_\mathit{fin}}{\assign_\mathit{fin}} \true$ always holds.
	\item Let $\varphi = \psi_1 \land \psi_2$. Let $t \in \mathbb{N}_0$. Then, if $\computation, t \TSLSatisfaction{\mathcal{V}}{\assign} \psi_1 \land \psi_2$, then both $\computation, t \TSLSatisfaction{\mathcal{V}}{\assign} \psi_1$ and $\computation, t \TSLSatisfaction{\mathcal{V}}{\assign} \psi_2$ hold. Thus, by induction hypothesis, we have $\computation_\mathit{fin}, t \TSLSatisfaction{\mathcal{V}_\mathit{fin}}{\assign_\mathit{fin}} \psi_1$ and $\computation_\mathit{fin}, t \TSLSatisfaction{\mathcal{V}_\mathit{fin}}{\assign_\mathit{fin}} \psi_2$ and therefore $\computation_\mathit{fin}, t\TSLSatisfaction{\mathcal{V}_\mathit{fin}}{\assign_\mathit{fin}} \psi_1 \land \psi_2$ follows.
	\item Let $\varphi = \psi_1 \Until \psi_2$. Let $t \in \mathbb{N}_0$. Then, if $\computation, t \TSLSatisfaction{\mathcal{V}}{\assign} \psi_1 \Until \psi_2$, then there exists a point in time $i \geq 0$ such that both $\computation, t+i \TSLSatisfaction{\mathcal{V}}{\assign} \psi_2$ and for all points in time~$j$ with $0 \leq j < i$, $\computation, t+j \TSLSatisfaction{\mathcal{V}}{\assign} \psi_1$ holds. Thus, by induction hypothesis, there exists a point in time $i \geq 0$ such that both $\computation_\mathit{fin}, t+i \TSLSatisfaction{\mathcal{V}_\mathit{fin}}{\assign_\mathit{fin}} \psi_2$ and for all points in time~$j$ with $0 \leq j < i$, $\computation_\mathit{fin}, t+j \TSLSatisfaction{\mathcal{V}_\mathit{fin}}{\assign_\mathit{fin}} \psi_1$ holds. Therefore, $\computation_\mathit{fin}, t\TSLSatisfaction{\mathcal{V}_\mathit{fin}}{\assign_\mathit{fin}} \psi_1 \Until \psi_2$ follows.
	\item Let $\varphi = t_p$. Let $t \in \mathbb{N}_0$. Then, if $\computation, t \TSLSatisfaction{\mathcal{V}}{\assign} t_p$, then $\recursiveAssign{\etaFunction{\computation}{t}{t_p}} = \true$. By \Cref{lem:correct_construction_fin}, we have $\recursiveAssign{\etaFunction{\computation}{t}{t_p}} = \recursiveAssignFin{\etaFunction{\computation_\mathit{fin}}{t}{t_p}}$ and therefore $\recursiveAssignFin{\etaFunction{\computation_\mathit{fin}}{t}{t_p}} = \true$ follows. Hence, $\computation_\mathit{fin}, t \TSLSatisfaction{\mathcal{V}_\mathit{fin}}{\assign_\mathit{fin}} t_p$ holds.
	\item Let $\varphi = \update{\cell{c}}{t_f}$. Let $t \in \mathbb{N}_0$. Then, if $\computation, t \TSLSatisfaction{\mathcal{V}}{\assign} \update{\cell{c}}{t_f}$, then $\atPos{\computation}{t}(\cell{c}) = t_f$. Since $\cell{c}$ occurs in $\varphi$, \ie, $\cell{c} \not\in \cells^\mathit{new}$, and since $\update{\cell{c}}{t_f}$ occurs in $\varphi$, then, by definition of $\computation_\mathit{fin}$, we have $\atPos{\computation_\mathit{fin}}{t}(\cell{c}) = \atPos{\computation}{t}(\cell{c})$. Thus, since $\atPos{\computation}{t}(\cell{c}) = t_f$, $\atPos{\computation_\mathit{fin}}{t}(\cell{c}) = t_f$ follows. Hence, $\computation_\mathit{fin}, t \TSLSatisfaction{\mathcal{V}_\mathit{fin}}{\assign_\mathit{fin}} \update{\cell{c}}{t_f}$ holds.
\end{itemize}

To conclude, we have proven that $\computation_\mathit{fin} \TSLSatisfaction{\mathcal{V}_\mathit{fin}}{\assign_\mathit{fin}} \varphi \land \varphi_\mathit{fin}$ and $\finitary{\varphi \land \varphi_\mathit{fin}}{\computation_\mathit{fin}}$ hold. Hence, by \Cref{lem:canonical_model_tu}, $\varphi \land \varphi_\mathit{fin}$ is finitary satisfiable in $\tu$.


\subsection{Correctness of \Cref{alg:TSL_TU_SAT_idea} (Proof of \Cref{thm:correctness_algorithm})}

We prove the correctness of \Cref{alg:TSL_TU_SAT_idea} in three steps: First, we prove that if the algorithm terminates with \sat~on the TSL formula $\varphi$, then there exists a respective finitary execution satisfying $\varphi$.
Second, we show that if \Cref{alg:TSL_TU_SAT_idea} terminates with \unsat~on the TSL formula $\varphi$, then there is no finitary execution that satisfies $\varphi$.
Third, we combine both results to conclude the proof of \Cref{thm:correctness_algorithm}.

\begin{lemma}\label{lem:correctness_algorithm_sat}
	Let $\varphi$ be a TSL formula.
	If \Cref{alg:TSL_TU_SAT_idea} terminates on $\varphi$ with \sat, then there exists an execution $\symbolicExecution$ such that $\computation \TSLSatisfaction{\mathcal{V}}{\assign} \varphi$ and $\finitary{\varphi}{ctr}$.
\end{lemma}
\begin{proof}
	Let $\mathcal{B}$ be the BSA that \Cref{alg:TSL_TU_SAT_idea} constructs for $\varphi$. Assume that the algorithm terminates with \sat. Then, the SMT query is satisfiable and hence there exists some accepting run $\pref.\rec^\omega$ such that $\projection{\executionEffect{\pref.\rec}}{2}$ is consistent. Since $\mathcal{B}$ is finitary, $\pref.\rec^\omega$ has a unique, finitary computation $\computation$.
	Since the SMT query is satisfiable, we have \[ \forall \cell{c} \in \cells.~ \etaFunction{\computation}{\abs{\pref\,}}{\cell{c}} = \etaFunction{\computation}{\abs{\pref\,}+\abs{\rec}}{\cell{c}} \]
	by definition of the execution effect.
	In the theory of equality, function applications preserve equality. Therefore, the execution of identical transitions after equal terms preserves equality as well. Since $\pref.\rec^\omega$ is lasso-shaped, \ie, $\rec$ is repeated infinitely often, we thus obtain \[ \forall t > \abs{\pref\,}.~ \forall \cell{c} \in \cells.~ \etaFunction{\computation}{t}{\cell{c}} = \etaFunction{\computation}{t+\abs{\rec}}{\cell{c}}. \]
	
	We now show that the constraint trace $\constraintTrace$ of $\computation$ is consistent, \ie, that there is no predicate term $\predTerm \in \predTerms$ such that both $(\predTerm,\true)$ and $(\predTerm,\false)$ are contained in $\bigcup_{t \in \mathbb{N}_0}\atPos{\constraintTrace}{t}$. Since $\atPos{\constraintTrace}{t-1} \subseteq \atPos{\constraintTrace}{t}$ holds by construction of $\constraintTrace$, it suffices to show that for all $t \in \mathbb{N}_0$, there is not predicate term $\predTerm \in \predTerms$ with $(\predTerm,\true),(\predTerm,\false)\in\atPos{\constraintTrace}{t}$. We prove this claim by induction on the considered point in time $t$:
	\begin{itemize}
		\item Let $t \leq \abs{\pref.\rec}$. We have $\projection{\executionEffect{\pref.\rec}}{2} = \atPos{\constraintTrace}{\abs{\pref\,}+\abs{\rec}}$ by definition of the execution effect. Thus, since $\projection{\executionEffect{\pref.\rec}}{2}$ is consistent as explained above, $\atPos{\constraintTrace}{\abs{\pref\,}+\abs{\rec}}$ is consistent. Therefore, by construction of $\constraintTrace$, $\atPos{\constraintTrace}{t'}$ is consistent for all points in time $t'$ with $0 \leq t' \leq \abs{\pref\,}+\abs{\rec} = \abs{\pref.\rec}$ and thus, since $t \leq \abs{\pref.\rec}$ holds by assumption, for all points in time $t'$ with  $0 \leq t' \leq t$.
		\item Let $t > \abs{\pref.\rec}$. Towards a contradiction, suppose that there is a predicate term $\predTerm \in \predTerms$ with $(\predTerm,\true),(\predTerm,\false)\in\atPos{\constraintTrace}{t}$. By induction hypothesis, not both $(\predTerm,\true)$ and $(\predTerm,\false)$ can be contained in $\atPos{\constraintTrace}{t-1}$. Hence, at least one of them is contained in $\atPos{\constraintTrace}{t} \setminus \atPos{\constraintTrace}{t-1}$. Hence, by construction of $\constraintTrace$, at least one of them is contained in $\{(\etaFunction{\computation}{t-1}{\predTerm},\projection{c_{t-1}}{\predTerm'}) \mid \predTerm' \in \guards \}$, where $c$ is the run of $\mathcal{B}$, \ie, $c = \pref.\rec^\omega$. Therefore, $\predTerm = \etaFunction{\computation}{t-1}{\predTerm'}$ for some $\predTerm' \in \guards$. By assumption, $t>\abs{\pref.\rec}$ and hence $\etaFunction{\computation}{t-1}{\predTerm'} = \etaFunction{\computation}{t-1-\abs{rec}}{\predTerm'}$ holds. Therefore, $\predTerm = \etaFunction{\computation}{t-1-\abs{rec}}{\predTerm'}$. Thus, for $v$ with $(\predTerm,v)\in(\atPos{\constraintTrace}{t}\setminus\atPos{\constraintTrace}{t-1})$, we have $(\predTerm,v)\in\atPos{\constraintTrace}{t-\abs{rec}}$. But then $\atPos{\constraintTrace}{t-\abs{rec}}$ is inconsistent as $\abs{\rec}>0$, contradicting the induction hypothesis.
	\end{itemize}
	
    Hence, we have shown that $\constraintTrace$ is consistent. Therefore, we can define an assignment function $\assign$ such that $(\computation,\mathcal{V},\assign)$ is an execution for $\pref.\rec^\omega$ as follows: $\mathcal{V} := \funcTerms$, $\langle \cell{c} \rangle := \cell{c}$ for $\cell{c} \in \cells$, $\langle f \rangle := \lambda (x_1, \dots, x_n).~f(x_1, \dots, x_n)$ for $f \in \funcLiterals$, and $\langle p \rangle := \lambda (x_1, \dots, x_n). \left(p(x_1, \dots, x_n)\in\bigcup_{t\in\mathbb{N}_0}\atPos{\constraintTrace}{t}\right)$ for $p \in \predLiterals$.
    This implies $\recursiveAssign{\predTerm} := \true$ if, and only if $(\predTerm,\true)\in\bigcup_{t\in\mathbb{N}_0}\atPos{\constraintTrace}{t}$.
	Therefore, $\symbolicExecution \in \symbolicLanguage{\mathcal{B}}$ holds and thus, by \Cref{lem:tsl_to_bsa_fini}, we have both $\computation \TSLSatisfaction{\mathcal{V}}{\assign} \varphi$ and $\finitary{\varphi}{\computation}$.
\end{proof}

\begin{lemma}\label{lem:correctness_algorithm_unsat}
	Let $\varphi$ be a TSL formula.
	If \Cref{alg:TSL_TU_SAT_idea} terminates on $\varphi$ with \unsat, then for all executions $\symbolicExecution$ with $\finitary{\varphi}{ctr}$, we have $\computation \not\TSLSatisfaction{\mathcal{V}}{\assign} \varphi$.
\end{lemma}
\begin{proof}
	Let $\mathcal{B}$ be the BSA that \Cref{alg:TSL_TU_SAT_idea} constructs for $\varphi$. Let $c$ be an accepting run of $\mathcal{B}$. Assume that the algorithm terminates with \unsat. Then, $c \not \in \compBSA$ holds for the last computed set $\compBSA$ in \Cref{alg:TSL_TU_SAT_idea}. Hence, there exists a subword $c_p$ of~$c$ such that $\projection{\executionEffect{c_p}}{2}$ is not consistent. By definition of the execution effect, $\projection{\executionEffect{c_{pr}.c_p}}{2}$ is thus inconsistent for every additional prefix $c_{pr}$.
	Recall that finitary BSAs have unique computations and constraint traces for a run. Let $\computation$ and $\constraintTrace$ be these for $c$ and $\mathcal{B}$.
	Let $c_{pr}$ be a finite run such that $c_{pr}.c_p$ is a prefix of~$c$. Then, $\projection{\executionEffect{c_{pr}.c_p}}{2} = \atPos{\constraintTrace}{\abs{c_{pr}} + \abs{c_p}}$.
	Hence, $\constraintTrace$ is inconsistent and, since $\constraintTrace$ is the unique constraint trace of the only computation $\computation$ of $c$, there is no execution for $c$.
	Thus, no execution of $\mathcal{B}$ has a respective accepting run and hence $\symbolicLanguage{\mathcal{B}} = \emptyset$. Thus, by \Cref{lem:tsl_to_bsa_fini}, $\computation \not \TSLSatisfaction{\mathcal{V}}{\assign} \varphi$ for all executions $\symbolicExecution$ with $\finitary{\varphi}{\computation}$.
\end{proof}

Using these results, we can now show the correctness of \Cref{alg:TSL_TU_SAT_idea}.
First, assume that \Cref{alg:TSL_TU_SAT_idea} terminates with \sat~on the TSL formula $\varphi$. Then, by \Cref{lem:correctness_algorithm_sat}, there exists an execution $\symbolicExecution$ such that both $\computation \TSLSatisfaction{\mathcal{V}}{\assign} \varphi$ and $\finitary{\varphi}{\computation}$ holds. Hence, by \Cref{lem:canonical_model_tu}, $\varphi$ is finitary satisfiable in $\tu$.

Second, assume that \Cref{alg:TSL_TU_SAT_idea} terminates with \unsat~on the TSL formula $\varphi$. Then, by \Cref{lem:correctness_algorithm_unsat}, we have $\computation \TSLSatisfaction{\mathcal{V}}{\assign} \varphi$ for all executions $\symbolicExecution$ with $\finitary{\varphi}{\computation}$. That is, there is no  execution $\symbolicExecution$ with both $\computation \TSLSatisfaction{\mathcal{V}}{\assign} \varphi$ and $\finitary{\varphi}{\computation}$. Hence, by definition of TSL satisfiability, $\varphi$ is unsatisfiable in $\tu$.


\section{Undecidability of TSL modulo $\tu$ Satisfiability}\label{app:undecidability}

\subsection{Encoding Numbers in TSL modulo $\tu$ (Proof of \Cref{lem:tu_numbers})}

We construct $\varphi_\mathit{num}$ as follows: The formula consists of two conjuncts $\varphi_1$ and $\varphi_2$, \ie, $\varphi_\mathit{num} = \varphi_1 \land \varphi_2$. The first conjunct ensures the left-to-right direction of the claim, while the second one ensures the right-to-left direction.
	
	The first conjunct is defined by \[\varphi_1 := \update{e}{\encRes} \land \Next \Globally\left(\update{e}{\encInc(e)} \land e \encEq e\right).\] Clearly, if $\varphi_1$ is entailed by an execution, then $\encInc^n(\encRes) \encEq \encInc^n(\encRes)$ holds for all $n \in \mathbb{N}_0$ in all models. This concludes the left-to-right direction.	
	The second conjunct is defined by \[\varphi_2 := \update{x}{z} \land \update{b}{z} \land \Next \Globally (\varphi_\mathit{eq} \land \varphi_\mathit{neq}),\] where
	\begin{align*}
		\varphi_\mathit{eq} &:= \left( x \encEq b \right) \rightarrow \left( \update{x}{z} \land \update{b}{\encInc(b)} \land \neg\left( b \encEq \encInc(b) \right) \land \neg\left( \encInc(b) \encEq b \right) \right)\\
		\varphi_\mathit{neq} &:= \neg\left( x \encEq b \right) \rightarrow \left( \update{x}{\encInc(x)} \land \update{b}{b} \land \neg\left( x \encEq \encInc(b) \right) \land \neg\left( \encInc(b) \encEq x \right) \right) 
	\end{align*}
	First, note that by the structure of $\varphi_2$, cell $b$ can only contain terms of the form $\encInc^n(\encRes)$ for $n \in \mathbb{N}_0$. Furthermore, starting with $n=0$, $n$ increases in every execution entailing $\varphi_2$, \ie, the terms appear in the order $\encRes, \encInc(\encRes), \encInc^2(\encRes), \dots$.
	To prove the right-to-left direction of the claim, we show that every execution entailing $\varphi_2$ requires from its models that for all $a,b \in \mathbb{N}_0$, if $a \neq b$ holds, then $\neg(\encInc^a(\encRes) \encEq \encInc^b(\encRes))$. Since we consider infinite traces, it suffices to show that (1) If $b$ contains the value $\encInc^n(\encRes)$, then for all $k_1, k_2 \in \mathbb{N}_0$ with $0 \leq k_1, k_2 \leq n$, if $k_1 \neq k_2$, then $\neg(\encInc^{k_1}(\encRes) \encEq \encInc^{k_2}(\encRes))$ holds, and (2) for all $n \in \mathbb{N}_0$,~$b$ will eventually reach the value $\encInc^n(\encRes)$.
	We prove these claims by induction on~$n$:
	
	Let $n=0$: Since $k_1 = k_2$ for all $k_1, k_2 \in \mathbb{N}_0$ with $0 \leq k_1, k_2 \leq n$, claim~(1) follows directly. Since $b$ is initially set to $\encRes$ in $\varphi_2$, claim (2) holds as well.
	
	Assume that the claims holds for $n \in \mathbb{N}_0$ with $n>0$. Consider the very first point in time at which $b$ is set to~$\encInc^n(\encRes)$. 
	Since $n>0$, the definitions of $\varphi_\mathit{eq}$ and $\varphi_\mathit{neq}$ ensure that then $x \encEq b$ holds. Hence, $x$ is set to $\encRes$.
	By induction hypothesis, we have $\neg(\encInc^k(\encRes) \encEq \encInc^n(\encRes))$ for all $k \in \mathbb{N}_0$ with $0 \leq k < n$. Hence, $\neg(x \encEq b)$ holds until $x$ is set to $\encInc^n(\encRes)$. By construction of $\varphi_2$ and particularly $\varphi_\mathit{neq}$, $x$ is thus updated with $\encInc(x)$ until~$x$ reaches $\encInc^n(\encRes)$. During this update process, $b$ cannot change its value and the last two conjuncts of the conclusion of $\varphi_\mathit{eq}$ enforce that both $\neg(\encInc^k(\encRes) \encEq \encInc^{n+1}(\encRes))$ and $\neg(\encInc^{n+1}(\encRes) \encEq \encInc^k(\encRes))$ hold for all $k$ with $0 \leq k < n$.
	When $x$ reaches $\encInc^n(\encRes)$, $\varphi_1$ enforces that $x \encEq b$ holds. Thus, $\varphi_\mathit{eq}$ ensures that~$b$ is set to $\encInc^{n+1}(\encRes)$. Hence claim (2) holds for $n+1$. Furthermore, $\varphi_\mathit{eq}$ enforces that both $\neg(\encInc^n(\encRes) \encEq \encInc^{n+1}(\encRes))$ and $\neg(\encInc^{n+1}(\encRes) \encEq \encInc^n(\encRes))$ hold. Therefore, together with the previous predicate enforcements, we obtain that $\neg(\encInc^k(\encRes) \encEq \encInc^{n+1}(\encRes))$ and $\neg(\encInc^{n+1}(\encRes) \encEq \encInc^k(\encRes))$ hold for all $k$ with $0 \leq k \leq n$. Hence, together with the induction hypothesis, this yields that for all $k_1,k_2 \in \mathbb{N}_0$ with $0 \leq k_1, k_2 \leq n+1$, if $k_1 \neq k_2$, then $\neg(\encInc^{k_1}(\encRes) \encEq \encInc^{k_2}(\encRes))$, concluding the induction step.


\subsection{Reducing GOTO Programs to TSL Formulas (Proof of \Cref{thm:undecidability_tsl_modulo_tu})}

In this section, we reduce \goto~programs to TSL formulas. This is used to prove the undecidability of the satisfiability of a TSL formula modulo the theories of uninterpreted functions, equality, and Presburger arithmetic.

\paragraph*{\goto~programs.}
First, we introduce \goto~programs. Note that our definition slightly differs from classical ones. However, they are equivalent in the sense that \goto~programs defined according to one definition can be simulated by a \goto~program defined according to the other one. 
We will prove the equivalence of our definition to classical ones later in this section.

\begin{definition}[\gotoNSC]\label{def:goto}
	A \emph{\goto~program} is tuple $\mathcal{G} = (L, V, \lambda)$, where 
	\begin{itemize}
		\item $L = \{ l_0, l_1,\dots,l_m \}$ is an ordered set of locations, $l_0$ is called starting location and $l_m$ is called terminating location. Each execution of $\mathcal{G}$ starts in $l_0$ and terminates when $l_m$ is reached.
		\item $V = \{v_0,\dots,v_n\}$ is a set of variables over the natural numbers. Variable $v_0$ is used to provide inputs to $\mathcal{G}$.
		\item $\lambda: L\setminus\{l_m\} \rightarrow (\{\gotoifeq\} \times V \times V \times L) \cup (\{\inc,\res\} \times V)$ assigns a so-called action to each location. Since $l_m$ is the terminating, it does not have an action. If $\lambda(l_i) = (\inc,v_i)$, then variable $v_i$ is incremented and the action at location $l_{i+1}$ is executed. If $\lambda(l_i) = (\res,v_i)$, then variable $v_i$ is set to zero and the action at location $l_{i+1}$ is executed. If $\lambda(l_i) = (\gotoifeq,v_j,v_k,l_l)$, then, if $v_j = v_k$, the action at location $l_l$ is executed. Otherwise, the action at location $l_{i+1}$ is executed.
	\end{itemize}
\end{definition}

We denote \goto~programs as $l_0 : \lambda(l_0); ~\dots;~l_{m-1} : \lambda(l_m)$. We omit a location if the program does not contain a jump to this particular locations.

Note that \goto~programs as described in \Cref{def:goto} differ from classical ones. Yet, they are equivalent in the sense that actions in our definition can be simulated by classic \goto~programs and vice versa:

Classic \goto~programs have assignments like $v_i := v_j \pm c$ for some constant $c \in \mathbb{N}_0$ and unconditional jumps like $\goto~l$, which the \goto~programs according to \Cref{def:goto} do not have. In the following, we show how they can be simulated in our version of \goto~programs:
\begin{itemize}
\item   $\goto~l$ can be simulated by $(\rec,v_f);(\gotoifeq,v_f,v_f,l)$ for an unused variable $v_f$.
\item   $v_i := v_j$ can be simulated through 
            \[(\res, v_i); l_s: (\gotoifeq, v_i,v_j,l_n); (\inc, v_i); \goto~l_s; l_n: \dots\]
\item   A decrementation $(\dec,v_i)$, where decrementing zero yields zero, can be simulated by
		\begin{align*}
			&v_f := v_i; (\res,v_i); (\gotoifeq, v_i, v_f, l_n);\\
			&l_s: v_g := v_i; (\inc, v_g); (\gotoifeq, v_f, v_g, l_n); (\inc, v_i); \goto~l_s; l_n: \dots
		\end{align*}
        for unused variables $v_f, v_g$.
\item   $v_i := v_j \pm c$ can be simulated using $v_i := v_j$ and applying $c$-times $\inc$ or $\dec$ on
        $v_i$ for $+$ or $-$, respectively.
\end{itemize}

On the other hand, classic \goto~programs can only jump on an equality check with a constant, i.e., something like $\textsc{If}~v_i = c~\goto~l$ for some constant~$c$.
The action $(\gotoifeq, v_i, v_j, l)$ of our modified \goto~programs as defined in \Cref{def:goto} can be simulated through
	\begin{align*}
		&v_f := v_j - v_i; ~ v_g := v_i - v_j; ~\textsc{If}~v_f = 0~\goto~l_s;\\
		&\goto~l_n; ~l_s: \textsc{If}~v_g = 0~\goto~l; ~l_n: \dots 
	\end{align*}
where $v_i := v_j - v_k$ is syntactic sugar for 
	\begin{align*}
		&v_i := v_j + 0; ~ v_h := v_k; ~l_s: \textsc{If}~v_h = 0~\goto~l_m;\\
		&v_h := v_h - 1; ~ v_i := v_j - 1; ~ \goto~l_s; l_m: \dots
	\end{align*}
for unused variables $v_f$, $v_g$, and $v_h$.

The question, whether some program terminates on every possible input is called the \emph{universal halting problem}.
The universal halting problem of Turing machines is neither semi-decidable nor co-semi-decidable.
Turing machines and \goto~programs have equivalent computation power. Therefore, we obtain the following theorem:

\begin{theorem}[High Undecidability of the Universal Halting Problem of GOTO programs]\label{thm:high_undecidability_universal_halting_problem}
	Let
        \[ H_\goto := \{ \mathcal{G} \in \goto \mid \forall i \in \mathbb{N}_0.~\mathcal{G}~\text{terminates on}~i \} \]
    and let $\mathcal{G} \in \goto$. Checking whether $\mathcal{G} \in H_\goto$ holds highly undecidable, \ie, neither semi-decidable nor co-semi-decidable.
\end{theorem}


\paragraph*{Constructing TSL formulas from \goto~programs.}
Given a \goto~program $\mathcal{G}$, we construct a TSL formula $\tslgoto{T}$ from $\mathcal{G}$ for a theory $T$ as follows:
Let $\mathcal{G} = (L, V, \lambda)$ be a \goto~program where $L = \{l_0, l_1, \dots, l_m \}$ and $V = \{v_0, \dots, v_n\}$.
Let $\varphi^\mathit{encnum}_T$ be a TSL formula encoding numbers in theory $T$. In fact, assume that $\varphi^\mathit{encnum}_T$ ensures that an update of a cell with $z$ corresponds to resetting the cell to zero, the function~$\encInc$ corresponds to incrementation, and the predicate $\encEq$ corresponds to equality. For instance, for $\tu$, $\varphi_\mathit{num}$ is such a TSL formula.
This is used in $\tslgoto{T}$ to ``simulate'' $\mathcal{G}$ on different inputs.
Then, $\tslgoto{T}$ is a TSL-formula 
    over $\cells := V \cup \{l, i, x\}$, where $i, l, x \not\in V$, $\Sigma_F := \{f, z\} \cup L$, where $f,g$ are unary and $z, l_0, \dots, l_m$ are constants, and $\Sigma_P := \{\encEq\} \cup \{p_i~|~0 \leq i \leq m\}$ where $p_1, \dots, p_m$ are unary.
The formula $\tslgoto{T}$ is defined by $\tslgoto{T} = \varphi^\mathit{encnum}_T \land \bigwedge_{1 \leq j \leq 5} \varphi_j$, where:

\begin{description}
    \item[$\varphi_1$:] \[ \bigwedge_{0 \leq a \leq m} \left( p_a(l_a()) \land \bigwedge_{0 \leq b \leq m, a \neq b} \neg p_a(l_b()) \right)\]
    \item[$\varphi_2$:] \[ \Next \Globally \left( \bigwedge_{0 \leq k < m} \left(p_k(l) \to \upd{i}{i} \land \mathit{a}(l_k)\right)\right),\]
            where 
            \begin{equation*}\begin{array}{c}
                \mathit{a}(l_k) :=  
                \begin{cases}
                        \upd{l}{l_{k + 1}()} \land \upd{v_j}{f(v_j)} &\\ \land\bigwedge_{0 \leq k \leq n, k \neq j}\upd{v_k}{v_k} 
                        &\text{if}~\lambda(l_k) = (\inc,v_j)\\\\

                        \upd{l}{l_{k + 1}()} \land \upd{v_j}{z()} &\\ \land \bigwedge_{0 \leq k \leq n, k \neq j}\upd{v_k}{v_k}
                        &\text{if}~\lambda(l_k) = (\res,v_j)\\\\

                        (\neg (v_i \encEq v_j) \to \upd{l}{l_{k + 1}()})&\\ 
                        \land (v_i \encEq v_j \to \upd{l}{l_a()}) &\\ 
                        \land \bigwedge_{0 \leq k \leq n}\upd{v_k}{v_k}
                        &\text{if}~\lambda(l_k) = (\gotoifeq,v_i, v_j, l_a) \\
            \end{cases}
            \end{array}\end{equation*}
        \item[$\varphi_3$:] \[ \Next\Globally \left(\!p_m(l) \to \upd{i}{f(i)} \land \upd{l}{l_0()} 
            \land \update{v_0}{i} \land \!\left( \bigwedge_{1 \leq k \leq n}\!\upd{v_k}{z()} \right) \!\!\right)\]
        \item [$\varphi_4$:]\[ \upd{i}{z()} \land \upd{l}{l_m()} \]
        \item [$\varphi_5$:]\[ \Globally\Eventually p_m(l) \]
 \end{description}

The cell $l$ is used to simulate the program counter of $\mathcal{G}$. 
The cells $v_1, \dots, v_n$ simulate the different variables of $\mathcal{G}$.
If $\varphi_1$ is satisfied, it implies that the constants $l_0, \dots, l_m$ are pairwise distinct and can be used to simulate the locations by checking if their respective predicate holds.
If satisfied, $\varphi_2$ simulates the execution of $\mathcal{G}$: $\varphi_2$ states that every time (the initial time step is an exception) the location is not the terminating location, the location and variables are modified according to the program execution of $\mathcal{G}$. More precisely: If the location's statement is to increase (reset) some variable, the variable is increased (reset) and the location is set to the next location. If the location's statement is an equality check, depending on whether the variables to check are equal or not, the location is either set to the target location or to the next one.

The cell $i$ is used to simulate the inputs given to $\mathcal{G}$. During the execution simulation, which is done by $\varphi_2$, $i$ is not modified. However, if the simulation reaches the halting location, $\varphi_3$ enforces that a new simulation of $\mathcal{G}$ on $i$ is started~--~by setting $l$ to the starting location, $v_0$ (the input variable) to $i$, and the other variables to zero -- and $i$ is incremented. In addition, $\varphi_4$ initially sets $i$ to zero and~$l$ to the halting location. Therefore, $\varphi_2$ and $\varphi_3$ have an initial $\Next$.  

This implies that $\mathcal{G}$ is first simulated on zero and then, if it terminates, on one, then two and so on.   
Lastly, $\varphi_5$ requires that $\mathcal{G}$ has to terminate infinitely often and in conclusion to indeed terminate on all inputs.
So overall, $\tslgoto{T}$ models that infinitely often $\mathcal{G}$ is simulated on some input (starting with zero)
    and if the simulation terminates, then the next simulation with the incremented input is done. The fact that $\varphi_5$ enforces that each of these simulations has to terminate, will allow us to show that the construction of $\tslgoto{T}$ indeed allows us to reduce $H_\goto$ to the satisfiability of a TSL formula in $T$:

\begin{lemma}\label{lem:correctness_translation_goto_tsl}
	Let $\mathcal{G}$ be a \goto~program and let $\tslgoto{T}$ be constructed as described above. Then, $\mathcal{G} \in H_\goto$ if, and only if, $\tslgoto{T}$ is satisfiable in $T$.
\end{lemma}
\begin{proof}
	First, let $\mathcal{G} \in H_\goto$. Then, for each $I \in \mathbb{N}$, $\mathcal{G}$ has a termination execution $ l_{t_{1,I}}l_{t_{2,I}} l_{t_{3,I}} \cdots  l_{t_{p_I,I}}$
    on $I$ such that both $l_{t_{1,I}} = l_0$ and $l_{t_{p_I, I}} = l_m$ hold. Hence,~$\varphi_\mathcal{G}$ is satisfiable in $T$.
    As described above, a satisfying execution first simulates an execution of $\mathcal{G}$ on zero ($z$ in the formula).
    Since, due to $\varphi_3$, this execution terminates, the executions of $\mathcal{G}$ on one, two, three, 
    \dots~are simulated as every execution terminates.
    Thus, the location cell of the computation has the form
        $l_m, l_{t_{1,0}} \cdots l_{t_{p_0,0}}, l_{t_{1,1}} \cdots l_{t_{p_1,1}}, l_{t_{1,2}} \cdots l_{t_{p_2,2}}, l_{t_{1,3}} \cdots $.
     With this, $\varphi_1$, $\varphi_2$, $\varphi_3$, and $\varphi_4$ are satisfied. As $\mathcal{G}$ always terminates, every simulation in the computation eventually reaches $l_m$ and so $\varphi_5$ is also satisfied.

	Second, let $\tslgoto{T}$ be satisfiable in $T$. Then, there exists some computation such that -- as described above --
    $\mathcal{G}$ is simulated on some input~$i$ (starting with zero); and if this simulation terminates, $\mathcal{G}$ is simulated on the incremented input.
    As $\tslgoto{T}$ also requires that every one of these executions terminate, $\mathcal{G}$ is simulated on every input (as the trace is infinitely long) and terminates on each of these inputs, i.e., $\mathcal{G} \in H_\goto$.
\end{proof}

With \Cref{lem:correctness_translation_goto_tsl} and \Cref{thm:high_undecidability_universal_halting_problem}, it now follows that the satisfiability of a TSL formula in a theory allowing for encoding numbers in TSL is highly undecidable:

\begin{theorem}\label{thm:high_undecidability}
	Let $T$ be a theory that allows for encoding resetting a variable to zero, incrementation, and equality with a TSL formula $\varphi^\mathit{encnum}_T$. Then, the satisfiability of a TSL formula in $T$ is highly undecidable.
\end{theorem}


\subsection{Undecidability of TSL modulo $\tp$ and $\te$ Satisfiability (Proofs of \Cref{thm:undecidability_tsl_modulo_tp,thm:undecidability_tsl_modulo_te})}

Using the reduction from the universal halting problem for \goto~programs to the satisfiability of a TSL formula, we can now prove the high undecidability of the satisfiability of a TSL formula in the theories of equality $\te$ and Presburger arithmetic $\tp$.
It remains to show that both theories allow for encoding resetting a variable to zero, incrementing, and equality.

These operations are already included in the theory of Presburger arithmetic. Hence, high undecidability of the satisfiability of a TSL formula modulo Presburger arithmetic directly follows with \Cref{thm:high_undecidability}.

For the theory of equality $\te$, we can encoding the required operations with a TSL formula $\varphi_\mathit{enc}$:
\begin{lemma}\label{lem:encoding_te}
	Let $f$ and $g$ be unary functions and let $z$ be a constant. Let $f^x(z)$ and $g^x(z)$ correspond to applying $f$ and $g$ $x$-times to $z$, respectively. There exists a TSL formula $\varphi_\mathit{enc}$ such that every execution entailing $\varphi_\mathit{enc}$ requires from its models that for all $a,b \in \mathbb{N}_0$, if $a \neq b$ holds, then $f^a(x) \neq f^b(z)$ holds.
\end{lemma}
\begin{proof}
	We construct $\varphi_\mathit{enc}$ as follows: \[\varphi_\mathit{enc} := \update{e}{z} \land \Next \Globally \left( \update{e}{f(z)} \land (e = g(f(e))) \land \neg (f(e) = z) \right) \]
	If $\varphi_\mathit{enc}$ is satisfied, then clearly $f^n(z) = g(f(f^n(z)))$ holds for all $n \in \mathbb{N}_0$ and thus we have $f^n(z) = g(f^{n+1}(z))$ for all $n \in \mathbb{N}_0$. Furthermore, for all $n \in \mathbb{N}_0$ with $n > 0$, $f^n(z) \neq z$ holds.
	
	Towards a contradiction, suppose that $a \neq b$ and $f^a(z) = f^b(z)$ hold for some $a,b \in \mathbb{N}_0$. Without loss of generality, let $a < b$. By the axioms of $\te$, we have $g^a(f^a(z)) = g^a(f^b(z))$ since $f^a(z) = f^b(z)$ holds by assumption. With $f^n(z) = g(f^{n+1}(z))$ for all $n \in \mathbb{N}_0$, it thus follows that $z = f^{b-a}(z)$ holds. Since $a < b$ and $a \neq b$ by assumption, we have $b-a>0$. Hence, $z = f^{b-a}(z)$ contradicts $f^n(z) \neq z$ for all $n \in \mathbb{N}_0$ with $n>0$.
\end{proof}
Hence, with \Cref{lem:encoding_te} and \Cref{thm:high_undecidability}, the high undecidability of the satisfiability of a TSL formula modulo the theory of equality follows.


\section{(Semi-)Decidable Fragments}\label{app:fragments}

\subsection{Reachability Formulas (Proof of \Cref{lem:semi-decidability_reachability})}

Since $\varphi$ is satisfiable, a BSA for $\varphi$ has an accepting state $q$ that is reachable by a consistent transition sequence $s$.
Since $\varphi$ lies in the reachability fragment, everything that is read after reaching $q$ does not influence the acceptance. In particular, only performing self-updates after reaching $q$ leads to an accepting lasso.
Such a self-update lasso satisfies the equality constraint generated by \Cref{alg:TSL_TU_SAT_idea}. Since the algorithm enumerates all lassos, this specific accepting lasso starting with $s$ and performing only self-updates afterwards will always be found.

\subsection{Formulas With a Single Top-Level Eventually-Operator (Proof of \Cref{lem:TSL-TU-SAT_termination_top-level_eventually})}

	Since $\varphi$ lies in the reachability fragment as well, the claim follows directly with \Cref{lem:semi-decidability_reachability} if $\varphi$ is finitary satisfiable. 
	Next, let $\varphi$ not be finitary satisfiable. By assumption, $\varphi$ is either time-bounded itself, or it is of the form $\Eventually \varphi'$, where~$\varphi'$ is time-bounded.
	The suitable BSA for $\varphi$ has a single accepting state $q_f$ with self-loops for all combinations of updates and predicate evaluations. It has a single initial state $q_0$. If $\varphi = \Eventually \varphi'$, then $q_0$ has the same kinds of self-loops as~$q_f$. In between, there might be intermediate states and edges, leading from $q_0$ to $q_f$ without any cycles. This corresponds to the inner time-bounded part of the formula.
	If $\varphi = \Eventually \varphi'$, then $\varphi'$ is unsatisfiable as well, \ie, the source of the unsatisfiability is always the time-bounded part.
	Hence, since there are only finitely many inner transition sequences as the inner part of the BSA is cycle-free, the exclusion mechanism of \Cref{alg:TSL_TU_SAT_idea} excludes all of them.
	Then, there is no path from the $q_0$ to $q_f$ anymore and hence \Cref{alg:TSL_TU_SAT_idea} terminates.

\subsection{Single-Cell Formulas (Proof of \Cref{lem:TSL-TU-SAT_termination_one_cell})}

Let $\varphi$ be a TSL formula such that $\abs{\cells} = 1$ with $\cells = \{x\}$ and assume that $\varphi$ is finitary satisfiable in $\tu$. Let $\mathcal{B} = (Q, Q_0, F, \fresh, \guards, \updateTerms, \delta)$ be the finitary BSA of $\varphi$ which is constructed in \Cref{alg:TSL_TU_SAT_idea}.
Then, by \Cref{thm:tsl_to_bsa}, $\theoryLanguage{\tu}{\mathcal{B}} \neq \emptyset$.
Hence, there exists some accepting run $c$ with an execution $\symbolicExecution$. Note that, since $\mathcal{B}$ is finitary, $\finitary{\varphi}{\computation}$ holds. To prove that \Cref{alg:TSL_TU_SAT_idea} terminates, we show that we can construct some lasso on $c$ that is consistent if the start and end of the recurrent part are equal.

Let $\operatorname{depth}: \funcTerms \to \mathbb{N}$ be the depth of the syntax tree of a given function or predicate term. 
Note that, as $\computation$ is finitary, the number of possible terms built up by $\computation$ with a maximum depth $d$, i.e.\ $\{\etaFunction{\computation}{t}{x} \mid t \in \mathbb{N}_0, \depth{\etaFunction{\computation}{t}{x}} = d\}$, is finite since for each new layer in the tree we only have a finite number of possibilities.
We define $maxDepth_{ctr}: \mathbb{N} \to \mathbb{N}$ which calculates the maximum depth of the term of $x$ in $\computation$ until a certain time step as follows:
        \[ \maxDepth{\computation}{t} := \max \{ \depth{\etaFunction{\computation}{k}{x}} \mid 0 \leq k \leq t \} \]
        
We distinguish two cases:

\begin{enumerate}
	\item There exists a bound $b$ such that for all $t \in \mathbb{N}_0$, $\maxDepth{\computation}{t} < b$ holds, \ie, there is a total bound $b$ for the depths of the terms of $x$. Then, the number of terms of $x$, i.e.\ $\{\etaFunction{\computation}{t}{x} \mid t \in \mathbb{N}_0\}$, is finite. In addition, the number of states in $\mathcal{B}$ is finite. Therefpre, by the pigeon-hole principle, there exist $i, j \in \mathbb{N}_0$, such that (1) $\etaFunction{\computation}{i}{x} = \etaFunction{\computation}{j}{x}$ and (2) $\projection{\atPos{c}{i}}{1} = \projection{\atPos{c}{j}}{1} \in F$ hold. Due to (1), there is a lasso. This lasso is consistent as the original execution effect with the same constraints is. Due to (2), the lasso visits an accepting state infinitely often and hence it is accepting. Thus, \Cref{alg:TSL_TU_SAT_idea} terminates.
   	 \item There is no total bound for the depths of the terms of $x$, \ie, for all $b$, there is a $t\in\mathbb{N}_0$ such that $\maxDepth{\computation}{t} > b$. Hence, the term of $x$ grows infinitely often. Let 
        \[M := \{ t + 1 \mid \depth{\etaFunction{\computation}{t+1}{x}} > \maxDepth{\computation}{t}, t \in \mathbb{N}_0 \}\]
    	be the set of time steps $t$ where $\etaFunction{\computation}{t}{x}$ is deeper than all previous seen terms.
   		Since the term of $x$ grows infinitely often, $\abs{M} = \infty$. Let 
	    \begin{align*}
	        \mathcal{T}_{col} :=  \{ t_p \in \predTerms \mid &\depth{t_p} \leq \max\{\depth{p} \mid p \in \allPredicates{\varphi}\}, \\
	                                & ~~ t_p~\text{has only symbols from}~\varphi \}
	    \end{align*}
    	be the set of all predicate terms with symbols from $\varphi$ that are not deeper than the predicates of $\varphi$. These predicate terms are, later in the proof, possible ``collisions'' of constraints. Note that $\mathcal{T}_{col}$ is finite.
    	
    	Let $\sim \subseteq \mathbb{N}_0 \times \mathbb{N}_0$ be an equivalence relation on time steps: $i \sim j$ holds if, and only if both $\projection{\atPos{c}{i}}{1} = \projection{\atPos{c}{j}}{1}$ and $\recursiveAssign{\etaFunction{\computation}{i}{t_p}} = \recursiveAssign{\etaFunction{\computation}{j}{t_p}}$ hold for all $t_p \in \mathcal{T}_{col}$.
    	Since $\mathcal{B}$ has only finitely many states and since $\mathcal{T}_{col}$ is finite, $\sim$ has only finitely many equivalence classes.
    	
    	Assume that we can find points in time $i, j \in \mathbb{N}_0$ with $i <  j$, such that $j \in M$, $i \sim j$ and there exists a $k$ with $i \leq k < j$ such that $\projection{\atPos{c}{k}}{1} \in F$, then $\pref.\rec^\omega$ is an accepting lasso, where $\pref = \atPos{c}{0}.\atPos{c}{1}\dots \atPos{c}{i - 1}$ and $\rec = \atPos{c}{i}.\atPos{c}{i + 1}\dots \atPos{c}{j - 1}$. Hence, $\pref.\rec^\omega$ causes $\mathit{SatSearch}$ and thus \Cref{alg:TSL_TU_SAT_idea} to terminate:
    	Since $i \sim j$ implies that $\projection{\atPos{c}{i}}{1} = \projection{\atPos{c}{j}}{1}$, $\pref.\rec^\omega$ is indeed a run of $\mathcal{B}$.
   		 In addition, since there is a $k$ with $i \leq k < j$ such that $\projection{\atPos{c}{k}}{1} \in F$, $\pref.\rec^\omega$ is accepting.
   		 Let $\constraintTrace$ be the constraint trace of $\computation$.
   		 Using the properties of the execution effect, it remains to show that, 
        \[ \psi_Q := \left( \bigwedge_{(t_p, v) \in \atPos{\constraintTrace}{j}} \begin{cases} t_p&\text{if}~v = true\\\neg t_p&\text{if}~v = \mathit{false}\end{cases} 
            \right)\land (\etaFunction{\computation}{i}{x} = \etaFunction{\computation}{j}{x})\]
   		 is satisfiable in $T_E$. 
   		 Note that $\atPos{\constraintTrace}{j}$ itself is already consistent.
         Since $j \in M$, $\etaFunction{\computation}{j}{x}$ is the deepest term until the time step $j$, \ie, we have \[ \depth{p} \leq \depth{\etaFunction{\computation}{j}{x}} + \max\{ \depth{t_p} \mid t_p \in \allPredicates{\varphi} \} \] for all $(p, v) \in P$, where $P := \projection{\executionEffect{\pref.\rec}}{2}$. Furthermore, no more than $\abs{\mathcal{T}_{col}}$ constraints in $P$ include $\etaFunction{\computation}{j}{x}$. Therefore, contradictions in $\psi_Q$ can only arise between predicate terms of the form $\{ \etaFunction{\computation}{i}{t_p} \mid t_p \in \mathcal{T}_{col}\}$ and predicate terms of the form $\{ \etaFunction{\computation}{j}{t_p} \mid t_p \in \mathcal{T}_{col}\}$. However, $i \sim j$ implies that $\recursiveAssign{\etaFunction{\computation}{i}{t_p}} = \recursiveAssign{\etaFunction{\computation}{j}{t_p}}$ holds for all $t_p \in \mathcal{T}_{col}$.
    	Therefore, $\psi_Q$ is satisfiable in $T_E$.
		
		Hence, to prove that \Cref{alg:TSL_TU_SAT_idea} terminates, it suffices to show that we can actually find $i,j \in \mathbb{N}_0$ with the above properties. Towards a contradiction, suppose there are not such points in time.
		Since $M$ is infinite, but $\sim$ has only finitely many equivalence classes, there exists an infinite set $M' \subseteq M$, such that $a \sim b$ holds for all $a, b \in M'$. In addition, let $m' := \min(M)$.
		Since there are no points in time with the above properties by assumption, for all $i, j \in M'$ with $ i < j$, there is no $k$ with $i \leq k < j$ such that $\projection{\atPos{c}{k}}{1} \in F$. Thus, since $M'$ is infinite and unbounded (by setting $i:= m'$), for all $b \in \mathbb{N}_0$, there is a $j \in M'$ with $j>k$ such that there is no $k$ with $m' \leq k < j$ and $\projection{\atPos{c}{k}}{1} \in F$. Therefore, there is a point in time $m' \in \mathbb{N}_0$ such that for all $b > m'$ we have $\projection{\atPos{c}{b}}{1} \not\in F$. This is a contradiction to the fact that $c$ is accepting.
\end{enumerate}


\section{Extended Description of the Algorithm for Checking TSL-$\tu$-SAT}\label{app:fulldescription}

The presentation of the algorithm for checking TSL-$\tu$-SAT in \Cref{alg:TSL_TU_SAT_idea} ignores how to realize the enumeration of accepting loops, the enumeration of finite subwords of runs, and the implementation of the infinite set~$\compBSA$. 
In this section, we present an extended algorithm description that addresses these issues.

\begin{algorithm}[t]
    \SetKwFunction{SMT}{SMT} 
    \SetKwProg{Fn}{Function}{}{}
    \SetKwProg{Rec}{record}{}{}
 
    \Rec{Block}{
        state $\in Q$;
        parent $\in \{$ Block$\times \delta$, $\bot\}$
    }
    \Fn{Trace (b:Block): $\delta^*$}{
        \eIf{b.parent = $\bot$}{\Return $\varepsilon$}{(parentBlock, trans) := b.parent\; \Return (Trace(parentBlock)).trans}
    }
    \Fn{FindAcceptingLoops(b:Block): Set(Block)}{
        S := $\emptyset$\;
        p := b.parent\;
        seenAcceptingState := false\;
        \While{p $\neq\bot$}{(b', \_) := p\; \If{b'.state $\in F$}{seenAcceptingState := true}\If{b'.state = b.state $\land$ seenAcceptingState}{S := S.add(b')}
        p := b'.parent\;
        }
        \Return S
    }
    \Fn{Successors(b:Block): Set(Block)}{
        S := $\emptyset$\;
        \For{trans $\in \delta$, b.state = $\pi_1$(trans)} {
            S.add(Block\{state := $\pi_4$(trans), parent := (b, trans)\})
        }
        \Return S
    }
    \caption{Tree Data Structure}\label{alg:block}
\end{algorithm}

The enumeration is done with a breadth-first search in a tree-shaped data structure.
This data structure can be found in~\Cref{alg:block}. 
\emph{Block} is a node of this data structure. 
It consists of the current state and a parent. A parent consists of the parent node and the transition in the BSA from the parent node to the current one. 
\emph{Trace} computes the finite run that is represented by some \emph{Block}.
This means that all \emph{Blocks} together represent a set of trees, where each path from a root to some node is a finite subword or prefix of a run.
\emph{FindAcceptingLoops} finds all previous \emph{Blocks} which are in the same state as the argument such that there is an accepting state in between. That is, \emph{FindAcceptingLoops} searches for all possible lasso-shaped runs that ``end'' in the current \emph{Block}.
\emph{Successors} computes all successor \emph{Blocks} of the current \emph{Block}, \ie the \emph{Blocks} with states that can be reached by one transition.
These \emph{Blocks} have as parent the respective transition and the current \emph{Block}.

The set of all accepting runs of $\mathcal{B}$, \ie, $\compBSA \cap \accBSA{\mathcal{B}}$ is realized by a Büchi automaton $\mathcal{B}_\text{Büchi} = (Q_B, \Sigma_B, \Delta, Q_{0B}, F_B)$ over the transitions of the BSA $\mathcal{B}$ for~$\varphi$ and an exclusion set $\mathcal{E}$ of finite words over transitions.
To construct $\mathcal{B}_\text{Büchi}$ it suffices to initialize it with the same structure as $\mathcal{B}$ (i.e.\ with the same states and transitions, but the original transitions of $\mathcal{B}$ are used as transition labels).
All transition sequences with a subword in $\mathcal{E}$ should be excluded. 
To this end, we build a Büchi automaton 
    \[\mathcal{B}^\mathcal{E}_\text{Büchi} := (Q_B \times \mathcal{P}(\Sigma_B^*), \Sigma_B, \Delta_\mathcal{E}, Q_{0B} \times \{\emptyset \}, F_{0B} \times \mathcal{P}(\Sigma_B^*))\]
with
    \[\Delta_\mathcal{E} := \{ ((q, e), t, (q', e')) \mid (q,t,q') \in \Delta, e' = \{ w \mid t.w \in \mathcal{E} \cup e \} , \varepsilon \not\in e' \}\]
such that $\mathcal{L}(\mathcal{B}^\mathcal{E}_\text{Büchi}) = \mathcal{L}(\mathcal{B}_\text{Büchi}) \setminus \{ a.b.c \in \delta^\omega \mid b \in \mathcal{E}\}$.
Note that this construction can be done on-the-fly.

\begin{algorithm}[t]
    \SetKwFunction{SMT}{SMT} 
    \SetKwProg{Fn}{Function}{}{}
    \SetKwProg{Rec}{record}{}{}
     \Fn{Exclude}{
        Queue := Queue.new()\;
        \For{$q \in Q$}{Queue.append(Block\{state := $q$, parent := $\bot$\})}
        \While{$\neg$(Queue.empty)}{
            b := Queue.pop()\;
            (\_, $P$):= $\executionEffect{Trace(b)}$\;
            \eIf{$\exists t_p.~(t_p, \true), (t_p, \false) \in P$}{$\mathcal{E}$ := $\mathcal{E} \cup $ Trace(b)\;}{Queue.append(Successors(b))}
        }
    }
    \caption{Finite Transition Sequence Exclusion}\label{alg:exclude}
\end{algorithm}

\emph{Exclude} in \Cref{alg:exclude} realizes the exclusion of inconsistent finite runs. 
This is done by a breadth-first search in all \emph{Block} trees starting in any state. 
For some \emph{Block}, the finite run is the \emph{Trace}. 
If the execution effect of this subword is inconsistent, then the word is added to the exclusion set $\mathcal{E}$.
The \emph{Successors} of the respective \emph{Block} are only searched if the execution effect is consistent, as if it is not, all \emph{Traces} of the successors are already excluded.
Note that the depth of the tree of the searched \emph{Block} corresponds to the parameter $n$ of \emph{UnsatSearch} in \Cref{alg:TSL_TU_SAT_idea}.

\begin{algorithm}[t]
    \SetKwFunction{SMT}{SMT} 
    \SetKwProg{Fn}{Function}{}{}
    \SetKwProg{Rec}{record}{}{}

    \Fn{Search}{%
        Queue := Queue.new()\;
        \For{$q_0 \in Q_0$}{Queue.append(Block\{state := $q_0$, parent := $\bot$\})}
        \While{$\neg$(Queue.empty)}{%
            b := Queue.pop()\;
            Construct $\mathcal{B}^\mathcal{E}_\text{Büchi}$\;
            \If{$\exists w \in \delta^\omega.~$(Trace(b).$w) \in \mathcal{L}(\mathcal{B}^\mathcal{E}_\text{Büchi})$}{
            \For{b' $\in$ FindAcceptingLoops(b)}{%
                ($v_p$, \_):= $\executionEffect{Trace(b')}$\;
                ($v_r$, $P$):=$\executionEffect{Trace(b)}$\;
                \If{\SMT $\!\!\left(\!\!\left(\bigwedge\limits_{(t_p, v) \in P} \!\begin{cases} t_p&\!\text{if}~v = \true\\\neg t_p&\!\text{if}~v = \false\end{cases} \!\right) 
                    \!\land\! \bigwedge\limits_{\cell{c} \in \cells} v_p(\cell{c}) = v_r(\cell{c}) \!\!\right)\!=$ \sat}{%
                    \Return~\sat
                }%
            }
            Queue.append(Successors(b))%
            }
        }
        \Return~\unsat
    }
\caption{Search}\label{alg:search}
\end{algorithm}

\emph{Exclude} does neither satisfiability nor unsatisfiability checking. 
These are done by the method \emph{Search} of \Cref{alg:search} running in parallel to \emph{Exclude}. 
Similar to \emph{Exclude}, \emph{Search} does a breadth-first search on \emph{Blocks}. 
However, it starts only in accepting states. 
For some \emph{Block} $b$, it first checks whether it is possible to use the \emph{Trace} of $b$ as a still valid prefix.
If this is not possible, every run with the \emph{Trace} of $b$ as a prefix is inconsistent. 
Therefore, \emph{Search} does not consider $b$ or its successors as candidates for a satisfiable loop. 
If every \emph{Block} of the queue is not possible anymore, $\mathcal{B}_\text{Büchi}$ is empty and \emph{Search} returns \unsat. 
On the other hand, if $b$ is possible, \emph{FindAcceptingLoops} is used to calculate all lasso-shaped runs. 
For some possible beginning \emph{Block} of a lasso $b'$, in the context of \Cref{alg:TSL_TU_SAT_idea}, $\pref$ is the \emph{Trace} of $b'$ and $\pref.\rec$ is the \emph{Trace} of $b$. 
As in \Cref{alg:TSL_TU_SAT_idea}, if the SMT query is satisfiable, \emph{Search} returns \sat.

The top-level part of the algorithm is given \Cref{alg:full}.
There, the BSA $\mathcal{B}$ for $\varphi$ and $\mathcal{B}_\text{Büchi}$ are generated.
Then, \emph{Exclude} and \emph{Search} are run in parallel.
 
\begin{algorithm}[t]
    \SetKwFunction{ExecSymbolic}{ExecSymbolic} 
    \SetKwInOut{Input}{input}
    \SetKwInOut{Output}{output}

    \SetKwFunction{SMT}{SMT} 
    \SetKwProg{Fn}{Function}{}{}
    \SetKwProg{Rec}{record}{}{}
    \Input{$\varphi$: TSL Fromula}
    \Output{\sat, \unsat}
    $\mathcal{B}$ := Finitary BSA for $\varphi$ according to Theorem~\ref{thm:tsl_to_bsa}\;
    $(Q, Q_0, F, \fresh, \guards, \updateTerms, \delta) := \mathcal{B}$\;
    $\mathcal{B}_\text{Büchi}$ := Büchi automaton for the computations of $\mathcal{B}$\;
    $\mathcal{E}$ := $\emptyset$\;
    \Return parallel(Search, Exclude)
\caption{Algorithm for Checking TSL-$\tu$-SAT (Extended Description)}\label{alg:full}
\end{algorithm}

\section{Application Benchmarks}\label{app:other_benchmarks}

\paragraph*{Gamemodechooser Assumptions.}
In this benchmark, we verify assumptions in the \textit{Gamemodechooser} module of the Syntroids FPGA arcade game~\cite{GeierH0F19} which is fully specified in (and synthesized from) TSL. The assumptions $\spec{a}$ that we want to verify state mutual exclusion of different predicates on a stream that is generated by the module. They are used in the module itself as well as in two other modules that depend on the \emph{Gamemodechooser}. To do so, we state some facts~$\spec{fact}$ on constants. The module is modeled in $\spec{sys}$. It differs slightly from the original one as we had to "shift" the computation by one to set an initial value inside the formula (which is done otherwise for synthesis) and that we abstracted the computation of $\mathit{rot}$.
\allowdisplaybreaks
\begin{align*}
    \spec{fact} & := \Globally \left( \mathit{isSM}(\mathit{score}()) \land \lnot \mathit{isSM}(\mathit{radar}()) \land \lnot \mathit{isSM}(\mathit{cockpit}()) \land \right. \\
     &~~~~\lnot \mathit{isRM}(\mathit{score}()) \land \mathit{isRM}(\mathit{radar}()) \land \lnot \mathit{isRM}(\mathit{cockpit}()) \land\\
     &~~~~\left.\lnot \mathit{isCM}(\mathit{score}()) \land \lnot \mathit{isCM}(\mathit{radar}()) \land \mathit{isCM}(\mathit{cockpit}()) \right)\\
    \spec{sys} & := (\Globally \updappl{rot}{f}{rot}) \land \updconst{gamemode}{cockpit} \land \\
               &~~~~\Next\Globally (((\mathit{isSM} \mathit{gamemode}  \land \mathit{lt}(\mathit{rot}, \mathit{neg}(gms()))) \lor \\
               &~~~~~~~~~~~~(\mathit{isRM}(\mathit{gamemode}) \land \mathit{gt}(rot, gms())))\\
               &~~~~~~~~~~\leftrightarrow \updconst{gamemode}{cockpit}) \land \\
               &~~~~\Next\Globally (\mathit{isCM}(\mathit{gamemode}) \land \mathit{lt}(\mathit{rot}, \mathit{neg} (gms())) \\
               &~~~~~~~~~~ \leftrightarrow \updconst{gamemode}{radar})\\
               &~~~~\Next\Globally (\mathit{isCM}(\mathit{gamemode}) \land \mathit{gt}(\mathit{rot}, gms()) \leftrightarrow \updconst{gamemode}{score})\\
    \spec{a} & := \Next \Globally (\lnot (\mathit{isSM}(\mathit{gamemode})   \land \mathit{isRM}(\mathit{gamemode}))) \land\\
             &~~~~\Next\Globally (\lnot (\mathit{isRM}(\mathit{gamemode})   \land  \mathit{isCM}(\mathit{gamemode}))) \land\\
             &~~~~\Next\Globally (\lnot (\mathit{isCM}(\mathit{gamemode})  \land  \mathit{isSM}(\mathit{gamemode})))
\end{align*}
To prove these assumptions to be correct, we verified unsatisfiability of \[\spec{fact} \land \spec{sys} \land \lnot \spec{a}.\]

\paragraph*{Inductive Assumption.}
In this benchmark, we state an (inductive) assumption~$\spec{a}$ that we want to verify, \ie, we check for its validity. Using duality, we thus check whether $\neg \spec{a}$ is unsatisifiable:

\begin{align*}
    \spec{a} & := \predconst{zero} \land \updconst{x}{zero} \land (\Next\Globally \updappl{x}{\mathit{inc}}{x}) \land (\Globally (\predcell{x} \to p(\mathit{inc}(x)))) \\
             &~~~~~\to \Next\Globally p(x)
\end{align*}

\paragraph*{Invariant Holding.}
In this benchmark, we consider a property $p$ that is invariant under certain changes ($\spec{inv}$). The system that we want to check does one of these changes non-deterministically ($\spec{sys}$) under a changing input ($\spec{inp}$) while the property holds initially. We want to check that the property indeed always holds ($\spec{spec}$).
\begin{align*}
    \spec{inv} & := \Globally (p(x) \to p (f(x))) \land \Globally (p(x) \land p(y) \to p(g(x,y))) \\
    \spec{inp} & := \Globally (\update{y}{h(y)} \land p(y)) \\
    \spec{sys} & := p(x) \land \Globally (\update{x}{x} \lor \update{x}{f(x)} \lor \update{x}{g(x,y)}) \\
    \spec{spec} & := \Globally p(x) 
\end{align*}
We query satisfiability of $\spec{inv} \land \spec{inp} \land \spec{sys} \land \lnot \spec{spec}$.

\paragraph*{One Of Two.}
This benchmark describes a system with a single cell that non-deterministically updates this cell with one of two constant, or leaves it unchanged ($\spec{sys}$). These constants can be distinguished by respective predicates ($\spec{fact}$). We want to check whether one of these predicates always holds ($\spec{spec}$).
\begin{align*}
    \spec{fact} & := \Globally \left( \mathit{isMode}_1 (m_1()) \land \mathit{isMode}_2 (m_2()) \land\right. \\
    			&~~~~~~~\left.\lnot \mathit{isMode}_1 (m_2()) \land \lnot \mathit{isMode}_2 (m_1()) \right) \\
    \spec{sys} & := \upd{mode}{m_1()} \land \\
    			&~~~~~~~ \Next\Globally \left( \upd{mode}{m_1()} \lor \upd{mode}{m_2()} \lor \updcell{mode}{mode} \right) \\
    \spec{spec} & := \Next\Globally (\mathit{isMode}_1(mode) \lor \mathit{isMode}_2(mode))
\end{align*}
To show this, we query satisfiability of $\spec{fact} \land \spec{sys} \land \lnot \spec{spec}$.

\paragraph*{One Of Three.} 
This benchmark is essentially the same as the previous benchmark \textit{One Of Two} but with three instead of two constants:
\begin{align*}
    \spec{fact} & := \Globally \left( \bigwedge_{i \in \{1,2,3\}} \mathit{isMode}_i (m_i()) \land \bigwedge_{i, j \in \{1,2,3\}, i \neq j} \lnot \mathit{isMode}_i (m_j()) \right) \\
    \spec{sys} & := \upd{mode}{m_1()} \land \Next\Globally \left( \right. \\ &~~~~ \left. \upd{mode}{m_1()} \lor \upd{mode}{m_2()} \lor \right. \\ &~~~~ \left. \upd{mode}{m_3()} \lor \updcell{mode}{mode} \right) \\
    \spec{spec} & := \Next\Globally (\mathit{isMode}_1(mode) \lor \mathit{isMode}_2(mode) \lor \mathit{isMode}_3(mode))
\end{align*}
We query satisfiability of $\spec{fact} \land \spec{sys} \land \lnot \spec{spec}$.

\paragraph*{Scheduler.}
This benchmark describes a scheduler that ``executes'' two programs represented by functions $f_1$, $f_2$. It has one memory cell and one cell to store the content for the program that is idle ($\spec{sys}$). We know that both programs hold some invariant property $p$ throughout each execution step ($\spec{inv}$). We want to check whether this property then always holds in both memory cells ($\spec{spec}$).
\begin{align*}
    \spec{init} & := \updconst{mem}{init} \land \updconst{mem}{init} \\
    \spec{run}  & := \Next\Globally (\updcell{swap}{swap} \lor \updcell{swap}{mem}) \land \\
                &~~~~\Next\Globally (\updcell{mem}{swap} \lor \updappl{mem}{f_1}{mem} \lor \updappl{mem}{f_2}{mem}) \\
    \spec{safe} & := \Globally (\updappl{mem}{f_1}{mem} \lor \updappl{mem}{f_2}{mem} \to \updcell{swap}{swap}) \land \\
                &~~~~\Globally (\updcell{mem}{swap} \leftrightarrow \updcell{swap}{mem}) \\
                &~~~~\Globally (\updappl{mem}{f_1}{mem} \land \Next\Next \updappl{mem}{f_2}{mem} \to \Next \updcell{swap}{mem}) \\
                &~~~~\Globally (\updappl{mem}{f_2}{mem} \land \Next\Next \updappl{mem}{f_1}{mem} \to \Next \updcell{swap}{mem}) \\
    \spec{live} & := \Globally\Eventually \updappl{mem}{f_1}{mem} \land \Globally\Eventually \updappl{mem}{f_2}{mem} \\
    \spec{sys}  & := \spec{init} \land \spec{run} \land \spec{safe} \land \spec{live} \\
    \spec{inv}  & := \Globally \left(p(\mathit{init}) \land (p(\mathit{mem}) \leftrightarrow p(f_1(\mathit{mem}))) \land (p(\mathit{mem}) \leftrightarrow p(f_2(\mathit{mem}))) \right)\\
    \spec{spec} & := \Next\Globally (p(\mathit{mem}) \land p(\mathit{swap}))
\end{align*}
We query satisfiability of $\spec{sys} \land \spec{inv} \land \lnot \spec{spec}$.

\paragraph*{Injector.}
This benchmark describes a system ($\spec{sys}$) that injects values with some property $p$ into a stream of data fulfilling $p$. The inputs are modeled by~$\spec{inp}$. We want to verify that the output values of the system still satisfy $p$ ($\spec{spec}$).
\begin{align*}
    \spec{inp} & := \Globally (\updappl{stream}{f}{stream} \land \updappl{inject}{f}{inject} \land \predcell{stream}) \\
    \spec{sys} & := \Globally (\updcell{memory}{inject} \leftrightarrow \predcell{inject}) \\
               &~~~~\Globally (\predcell{inject} \to \Next\Eventually \updcell{out}{memory})\\
               &~~~~\Globally (\updcell{out}{memory} \lor \updcell{out}{stream})\\
    \spec{spec} &:= \Next\Globally \predcell{out}
\end{align*}
We query satisfiability of $\spec{inp} \land \spec{sys} \land \lnot \spec{spec}$.

\paragraph*{Holding Arbiter.}
This benchmark describes a system consisting of two equal client systems $\spec{sys1}$, $\spec{sys2}$ and an arbiter $\spec{arb}$ which copies data from the client systems on request. The client systems only pose requests if their data fulfills a property $p$. We check that the arbiter only outputs data with this property~($\spec{spec}$). $\spec{inp}$ is used for input simulation and $\spec{fact}$ models constant facts.
\begin{align*}
    \spec{inp}  & := \Globally \updappl{in1}{f}{in1} \land \Globally \updappl{in2}{f}{in2} \\
    \spec{fact} & := \Globally \predconst{default}\\
    \spec{sys1} & := \Globally (\predcell{in1} \to \updcell{out1}{in1} \land \updconst{req1}{on}) \land\\
                &~~~~\Globally (\lnot \predcell{in1} \to \updcell{out1}{out1} \land \updconst{req1}{off})\\
    \spec{sys2} & := \Globally (\predcell{in2} \to \updcell{out2}{in2} \land \updconst{req2}{on}) \land\\
                &~~~~\Globally (\lnot \predcell{in2} \to \updcell{out2}{out2} \land \updconst{req2}{off})\\
    \spec{arb}  & := \updconst{out}{default} \land  \Next\Globally \left( \right.\\
                &~~~~~~(\updconst{req1}{on} \to \Next\Eventually \updcell{out}{out1}) \land\\ 
                &~~~~~\left.(\updconst{req2}{on} \to \Next\Eventually \updcell{out}{out2}) \right)\\
    \spec{spec} & := \Next\Globally \predcell{out}\\
\end{align*}
We query satisfiability of $\spec{inp} \land \spec{fact} \land (\spec{sys1} \land \spec{sys2} \land \spec{arb}) \land \lnot \spec{spec}$.

\paragraph*{Small Holding Arbiter.}
This benchmark is a variation of the previous one where the systems are replaced by a formula describing their observed behavior.
\begin{align*}
    \spec{inp}  & := \Globally \updappl{in1}{f}{in1} \land \Globally \updappl{in2}{f}{in2} \\
    \spec{fact} & := \Globally \predconst{default}\\
    \spec{sys1} & := \Globally \left( \updconst{req1}{on} \to (\updcell{out}{out1}~\mathcal{R}~\predcell{out}) \right) \\
    \spec{sys2} & := \Globally \left( \updconst{req2}{on} \to (\updcell{out}{out2}~\mathcal{R}~\predcell{out}) \right) \\
    \spec{arb}  & := \updconst{out}{default} \land  \Next\Globally \left( \right.\\
                &~~~~~~(\updconst{req1}{on} \to \Eventually \updcell{out}{out1}) \land\\ 
                &~~~~~\left.(\updconst{req2}{on} \to \Eventually \updcell{out}{out2}) \right)\\
    \spec{spec} & := \Next\Globally \predcell{out}
\end{align*}

\paragraph*{Pass Through Arbiter.}
This benchmark describes an arbiter $\spec{sys}$ which copies data from two input streams if it fulfills a property $p$ and schedules it to be written to a single output stream (such that no stream is left out). The ``requests'' from a classic arbiter are the fulfilled properties. Again, we want to check whether the outputted data fulfills $p$~($\spec{spec}$). In addition, we use $\spec{inp}$ to simulate inputs and $\spec{init}$ to describe the initial state.
\begin{align*}
    \spec{inp} & := \Globally \updappl{in1}{f}{in1} \land \Globally \updappl{in2}{f}{in2} \\
    \spec{sys} & := \Globally (\predcell{c1} \to \Eventually \updcell{out}{c1}) \land\\
               &~~~~\Globally (\predcell{c2} \to \Eventually \updcell{out}{c2}) \land\\
               &~~~~\Globally (\updcell{out}{c1} \to \predcell{c1}) \land \Globally (\updcell{out}{c2} \to \predcell{c2})\land\\
               &~~~~\Globally (\updcell{c1}{in1} \leftrightarrow \predcell{in1}) \land \Globally (\updcell{c2}{in2} \leftrightarrow \predcell{in2})\\
    \spec{init} & := \predcell{c1} \land \predcell{c2} \land \predcell{out}\\
    \spec{spec} & := \Next\Globally \predcell{out}
\end{align*}
We query satisfiability of $\spec{inp} \land \spec{sys} \land \spec{init} \land \lnot \spec{spec}$.

\paragraph*{Approximate Pass Through Arbiter.}
This benchmark is a variation of the previous one where the granting behavior is replaced by a non-deterministic choice.
\begin{align*}
    \spec{inp} & := \Globally \updappl{in1}{f}{in1} \land \Globally \updappl{in2}{f}{in2} \\
    \spec{sys} & := \Globally (\updcell{out}{c1} \lor \updcell{out}{c2} \lor \updcell{out}{out}) \land\\
               &~~~~\Globally (\updcell{out}{c1} \to \predcell{c1}) \land \Globally (\updcell{out}{c2} \to \predcell{c2}) \\
               &~~~~\Globally (\updcell{c1}{in1} \leftrightarrow \predcell{in1}) \land \Globally (\updcell{c2}{in2} \leftrightarrow \predcell{in2})\\
    \spec{init} & := \predcell{c1} \land \predcell{c2} \land \predcell{out}\\
    \spec{spec} & := \Next\Globally \predcell{out}
\end{align*}

\end{document}